\documentclass[10pt,twocolumn,journal]{IEEEtran}

\usepackage{amsmath, mathrsfs}

\makeatletter
\def\maketag@@@#1{\hbox{\m@th\normalfont\normalsize#1}}
\makeatother

\allowdisplaybreaks

\usepackage{amsfonts}
\usepackage{amssymb}
\usepackage{color}
\usepackage{multirow}
\usepackage{graphicx}

\usepackage{caption}
\usepackage{subcaption}
\usepackage[numbers,sort&compress]{natbib}
\usepackage{balance}

\def\mindex#1{\index{#1}}

\def\sq{\hbox{\rlap{$\sqcap$}$\sqcup$}}
\def\qed{\ifmmode\sq\else{\unskip\nobreak\hfil
\penalty50\hskip1em\null\nobreak\hfil\sq
\parfillskip=0pt\finalhyphendemerits=0\endgraf}\fi\medskip}

\long\def\defbox#1{\framebox[.9\hsize][c]{\parbox{.85\hsize}{%
\parindent=0pt
\baselineskip=12pt plus .1pt      
\parskip=6pt plus 1.5pt minus 1pt 
 #1}}}

\long\def\beginbox#1\endbox{\subsection*{}%
\hbox{\hspace{.05\hsize}\defbox{\medskip#1\bigskip}}%
\subsection*{}}

\def\endbox{}

\def\diag{{\text{diag}}}

\def\tr{\mathsf{Tr}}

\newsavebox{\junk}
\savebox{\junk}[1.6mm]{\hbox{$|\!|\!|$}}

\def\argmin{\mathop{\rm arg\, min}}

\def\argmax{\mathop{\rm arg\, max}}

\newcommand{\field}[1]{\mathbb{#1}}

\def\ZZ{\field{Z}}

\def\bC{{\mathbb C}}

\def\bE{{\mathbb E}}

\def\bH{{\mathbb H}}

\def\bT{{\mathbb T}}

\def\bV{{\mathbb V}}

\def\bfA{{\bf A}}
\def\bfB{{\bf B}}
\def\bfC{{\bf C}}
\def\bfD{{\bf D}}

\def\bfH{{\bf H}}
\def\bfI{{\bf I}}

\def\bfK{{\bf K}}
\def\bfL{{\bf L}}
\def\bfM{{\bf M}}
\def\bfN{{\bf N}}

\def\bfR{{\bf R}}
\def\bfS{{\bf S}}
\def\bfT{{\bf T}}
\def\bfU{{\bf U}}
\def\bfV{{\bf V}}
\def\bfW{{\bf W}}
\def\bfX{{\bf X}}
\def\bfY{{\bf Y}}
\def\bfZ{{\bf Z}}

\def\bfa{{\bf a}}

\def\bfd{{\bf d}}
\def\bfe{{\bf e}}
\def\bff{{\bf f}}
\def\bfg{{\bf g}}
\def\bfh{{\bf h}}

\def\bfm{{\bf m}}
\def\bfn{{\bf n}}

\def\bfv{{\bf v}}
\def\bfw{{\bf w}}
\def\bfx{{\bf x}}
\def\bfy{{\bf y}}

\def\sfH{{\sf H}}

\def\sfN{{\sf N}}

\def\sfa{{\sf a}}

\def\sfr{{\sf r}}

\def\bfmath#1{{\mathchoice{\mbox{\boldmath$#1$}}%
{\mbox{\boldmath$#1$}}%
{\mbox{\boldmath$\scriptstyle#1$}}%
{\mbox{\boldmath$\scriptscriptstyle#1$}}}}

\def\bfmY{\bfmath{Y}}

\def\bfmhhaY{\bfmath{\hhaY}} 
\def\bfmhhaY{\hbox to 0pt{$\widehat{\bfmY}$\hss}\widehat{\phantom{\raise 1.25pt\hbox{$\bfmY$}}}}

\def\til={{\widetilde =}}

\def\clA{{\cal A}}

\def\clC{{\cal C}}
\def\clD{{\cal D}}

\def\clG{{\cal G}}

\def\clK{{\cal K}}

\def\clN{{\cal N}}

\def\clU{{\cal U}}

\def\clX{{\cal X}}

\def\Var{{\mathbb{V}\sfa \sfr}}

 \def\FRAC#1#2#3{\genfrac{}{}{}{#1}{#2}{#3}}

\def\ddtp{{\mathchoice{\FRAC{1}{d^{\hbox to 2pt{\rm\tiny +\hss}}}{dt}}%
{\FRAC{1}{d^{\hbox to 2pt{\rm\tiny +\hss}}}{dt}}%
{\FRAC{3}{d^{\hbox to 2pt{\rm\tiny +\hss}}}{dt}}%
{\FRAC{3}{d^{\hbox to 2pt{\rm\tiny +\hss}}}{dt}}}}

\def\average#1,#2,{{1\over #2} \sum_{#1}^{#2}}

\def\eye(#1){{\bf(#1)}\quad}

\def\var{{\bV\sfa\sfr}}

\newtheorem{theorem}{{\bf Theorem}}

\newtheorem{proposition}{{\bf Proposition}}

\newtheorem{remark}{{\bf Remark}}

\def\eq#1/{(\ref{e:#1})}

\newcommand{\inp}[2]{{\langle #1, #2 \rangle}}

\newcommand{\beqn}[1]{\notes{#1}%
\begin{eqnarray} \elabel{#1}}

\newcommand{\eeqn}{\end{eqnarray} }

\newcommand{\beq}[1]{\notes{#1}%
\begin{equation}\elabel{#1}}

\newcommand{\eeq}{\end{equation}}

\def\bdes{\begin{description}}
\def\edes{\end{description}}

\newcounter{rmnum}

\newcounter{anum}

{\end{list}}

\def\ass(#1:#2){(#1\ref{#1:#2})}

\def\ritem#1{
\item[{\sf \ass(\current_model:#1)}]
}

\newenvironment{recall-ass}[1]{%
\begin{description}
\def\current_model{#1}}{
\end{description}
}

\long\def\comment#1{}

\newfont{\bbb}{msbm10 scaled 700}

\newfont{\bb}{msbm10 scaled 1100}

\newcommand{\EE}{\mbox{\bb E}}

\newcommand{\av}{{\bf a}}

\newcommand{\hv}{{\bf h}}

\newcommand{\mv}{{\bf m}}
\newcommand{\nv}{{\bf n}}

\newcommand{\xv}{{\bf x}}
\newcommand{\yv}{{\bf y}}

\newcommand{\Bm}{{\bf B}}
\newcommand{\Cm}{{\bf C}}
\newcommand{\Dm}{{\bf D}}

\newcommand{\Fm}{{\bf F}}

\newcommand{\Id}{{\bf I}}

\newcommand{\Mm}{{\bf M}}
\newcommand{\Nm}{{\bf N}}

\newcommand{\Sm}{{\bf S}}
\newcommand{\Tm}{{\bf T}}

\newcommand{\Wm}{{\bf W}}
\newcommand{\Vm}{{\bf V}}
\newcommand{\Xm}{{\bf X}}
\newcommand{\Ym}{{\bf Y}}

\newcommand{\Gammam}{\hbox{\boldmath$\Gamma$}}
\newcommand{\Lambdam}{\hbox{\boldmath$\Lambda$}}

\newcommand{\Sigmam}{\hbox{\boldmath$\Sigma$}}

\newcommand{\trace}{{\hbox{tr}}}

\newcommand{\transp}{{\sf T}}

\usepackage{tikz}
\usepackage{pgfplots}
\pgfplotsset{compat=newest}
\usetikzlibrary{patterns}
\usetikzlibrary{positioning}
\usetikzlibrary{datavisualization}
\usetikzlibrary{datavisualization.formats.functions}
\usetikzlibrary{backgrounds}
\usetikzlibrary{shapes,snakes}

\usepgfplotslibrary{external} 
\def\herm{{\sfH}}
\def\sdet{{\mathsf{det}}}
\def\snr{{\mathsf{snr}}}
\def\sinr{{\mathsf{sinr}}}
\def\sub{{\mathsf{sub}}}
\def\sp{{\mathsf{span}}}

\newcommand\norm[2]{{\|#1\|_{#2}}}

\newcommand{\inpb}[2]{{\langle #1,#2 \rangle_{\bfB}}}

\def\cg{{\clC\clN}} 
\def\gb{{GB}}
\def\gl{{GL}}
\def\Hproj{\underline{\mathbf{\sfH}}}
\def\Hmat{\underline{\bfH}}

\def\algML{{\bf Algorithm 1}}
\def\nameML{{AML}}
\def\algRMMV{{\bf Algorithm 2}}
\def\nameRMMV{{RMMV}}
\def\algCOR{{\bf Algorithm 4}}
\def\nameCMP{{CMP}}

\def\algSR{{\bf Algorithm 3}}
\def\nameSR{{SR}}

\def\lcav{{L_{\mathsf{cav}}}}
\def\lvex{{L_{\mathsf{vex}}}}

\begin{document}

\title{Massive MIMO Channel Subspace Estimation from Low-Dimensional Projections}
\author{Saeid Haghighatshoar,  \IEEEmembership{Member, IEEE,} Giuseppe Caire,
\IEEEmembership{Fellow, IEEE}%
\thanks{A shorter version of this paper was presented at the International Zurich Seminar, Zurich, Switzerland, March 2016.}
\thanks{The authors are with the Communications and Information Theory Group, Technische Universit\"{a}t Berlin (\{saeid.haghighatshoar, caire\}@tu-berlin.de).}\vspace{-4mm}
}

\maketitle

\begin{abstract}
Massive MIMO is a variant of multiuser MIMO where the number of base-station antennas $M$ is very large (typically $\approx 100$), 
and generally much larger than the number of spatially multiplexed data streams (typically $\approx 10$). 
The benefits of such approach have been intensively investigated
in the past few years, and all-digital experimental implementations have also been demonstrated.
Unfortunately, the front-end A/D conversion necessary to drive hundreds of antennas, with a signal bandwidth of the order of 
10 to 100 MHz, requires very large sampling bit-rate and power consumption.

In order to reduce such implementation requirements, 
Hybrid Digital-Analog architectures have been proposed.
In particular, our work in this paper is motivated by one of such schemes named
{\em Joint Spatial Division and Multiplexing} (JSDM), where the downlink precoder  (resp., uplink linear receiver) 
is split into the product of a baseband linear projection (digital) and an RF reconfigurable beamforming network (analog), 
such that only a reduced number $m \ll M$ of A/D converters and RF modulation/demodulation chains is needed. 
In JSDM, users are grouped according to similarity of their channel dominant subspaces, and these groups are separated by the analog beamforming stage, where multiplexing gain in each group is achieved using the digital precoder. Therefore, it is apparent that 
extracting the channel subspace information of the $M$-dim channel vectors from snapshots of $m$-dim 
{\em projections}, with $m \ll M$, plays a fundamental role in JSDM implementation. 

In this paper, we develop novel efficient algorithms that require sampling only $m=O(2 \sqrt{M})$ specific array 
elements according  to a coprime sampling scheme, and for a given $p \ll M$, return a $p$-dim beamformer that has 
a performance comparable with the best $p$-dim beamformer that can be designed from the full knowledge of the 
{\em exact} channel covariance matrix.  
{We assess the performance of our proposed estimators both analytically and empirically 
via numerical simulations. We also demonstrate by simulation that 
the proposed subspace estimation methods provide near-ideal performance 
for a massive MIMO JSDM system, by comparing with the case where the user channel 
covariances are perfectly known.}
\end{abstract}

\section{Introduction}  \label{sec:intro}

\PARstart{C}{onsider} a multiuser MIMO channel formed by a base-station (BS) with $M$ antennas and $K$ single-antenna mobile 
users in a cellular network. 
Following the current {\em massive MIMO} approach 
\cite{Marzetta-TWC10,Huh11,hoydis2013massive,larsson2014massive}, 
uplink (UL) and downlink (DL) are organized in Time Division Duplexing (TDD), and the BS transmit/receive hardware 
is designed or calibrated in order to preserve UL-DL reciprocity \cite{shepard2012argos,rogalin2014scalable} 
such that the BS can estimate the channel vectors of the users from UL training signals 
sent by the users on orthogonal dimensions.  Since there is no multiuser interference on the UL training phase, 
in this paper we shall focus on the basic channel estimation problem for a single user. 

In massive MIMO systems, the number of antennas $M$ is typically much larger than the number of users $K$ 
scheduled to communicate over a given transmission time slot (i.e., the number of spatially multiplexed data streams). 
Letting $D$ denote the duration of a time slot (expressed in channel uses), $\tau D$ channel uses for some $\tau\in(0,1)$, are dedicated to training and the remaining 
$(1 - \tau)D$ channel uses are devoted to data transmission, where it is assumed that $D$ is not larger than the 
channel coherence block length, i.e., the number of channel uses 
over which the channel is
nearly constant \cite{Marzetta-TWC10}.  It turns out that for isotropically distributed channel vectors
with $\min\{M, K\} \geq D/2$, it is optimal to devote a fraction $\tau = 1/2$ of the slot to channel 
estimation while serving only $D/2$ out of $K$ users in the remaining half \cite{Marzetta-TWC10}.\footnote{When $K > D/2$, then groups of $D/2$ users are scheduled
over different time slots such that all users achieve a positive throughput (i.e., rate averaged over a long sequence of scheduling slots).} 

In many relevant scenarios, the channel vectors are highly correlated since the propagation occurs through a small 
set of Angle of Arrivals (AoAs). This correlation can be exploited to improve the system multiplexing gain and decrease 
the training overhead.  A particularly effective scheme is the Joint Space Division and Multiplexing (JSDM) approach proposed and analyzed in
\cite{adhikary2013joint,nam2014joint,adhikary2014joint,adhikary2014massive,adhikary2014spatial}.
JSDM starts from the consideration that for a user with a channel vector $\bfh\in \bC^M$
the signal covariance matrix $\bfS=\bE[\bfh \bfh^\herm]$ is typically low-rank\footnote{This is especially true in the case of a tower-mounted BS and/or in the case of mm-wave channels, as experimentally confirmed by 
channel measurements (see \cite{adhikary2014joint} and references therein).}.
Moreover, according to the well-known and widely accepted Wide-Sense Stationary Uncorrelated Scattering (WSSUS) channel model, 
$\Sm$ is invariant over time and frequency. In particular, while the small-scale fading 
has a coherence time between 0.1s and 10\,ms for motion speed between 1\,m/s to 10\,m/s at the carrier frequency of 3 GHz,
the time over which the channel vector can be considered WSS is of the order of tens of seconds, i.e., 
from 2 to 4 orders of magnitude larger. Hence, estimating the  signal subspace of a user is a much easier task than estimating 
the instantaneous channel vector $\hv$ on each coherence time slot. 
This is especially important in mm-wave channels (e.g., carrier frequency of the order of 30 GHz) since, due to the higher carrier frequency, the Doppler bandwidth of these channels is large and therefore $D$ is small, i.e., the multiplexing gain of $D/2$ achieved by estimating the channels by TDD
on each given slot as in \cite{Marzetta-TWC10} is significantly impaired. 

When the subspace information for the users can be accurately estimated over a long sequence of time slots, JSDM partitions the users
into $G > 1$ groups such that users in each group have approximately the same dominant channel subspace 
\cite{adhikary2013joint,nam2014joint,adhikary2014joint}. The overall multiplexing gain is obtained in two stages, as the concatenation of two
linear projections. Namely, groups are separated by zero-forcing beamforming that uses only the group subspace information. 
Then, additional multiuser multiplexing gain can be obtained by conventional linear precoding applied independently in each group. 
In this way, the system multiplexing gain can be boosted by $G$ such that a decrease in $D$ can be compensated 
by a larger $G$ \cite{DBLP:journals/corr/NamCKH15}.

Furthermore, JSDM lends itself naturally to a Hybrid Digital Analog (HDA) implementation, where the group-separating beamformer
can be implemented in the analog (RF) domain, and the multiuser precoding inside each group is implemented in the digital (baseband) domain. 
The analog beamforming projection reduces the dimensionality from $M$ to some intermediate dimension 
$m \ll M$.  Then, the resulting $m$ inputs (UL) are converted into digital baseband signals, and are further processed in the digital domain. This has the additional non-trivial advantage that only $m \ll M$ RF chains (A/D converters and modulators) are needed, thus reducing significantly the 
massive MIMO BS receiver/transmitter front-end complexity and power consumption. 

From what said, it is apparent that a central task at the BS side consists in estimating, for each user, 
a subspace containing a significant amount of its received signal power.  
Since in an HDA implementation we do not have direct access to 
all the $M$ antennas, but only to $m \ll M$ analog output observations, we need to estimate 
this subspace from snapshots of a low-dim 
projection of the signal.

%
\subsection{Contribution} \label{sec:contribution}

In this paper, we aim to design such a subspace estimator for a BS with a large uniform linear array (ULA) with $M \gg 1$ antennas. 
The geometry of the array is shown in Fig.~\ref{fig:sc_channel}, with array elements having uniform spacing $d$. 

\begin{figure}[h]
\centering
\includegraphics{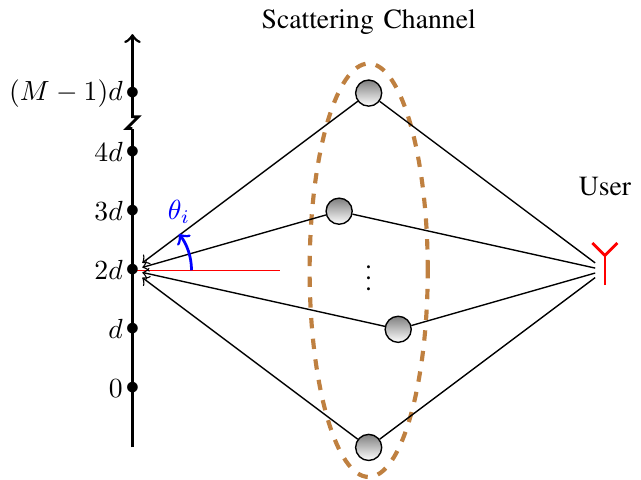}
\caption{{\small 
Array configuration in a multi-antenna receiver in the presence of a scattering channel with discrete angle of arrivals.}}
\label{fig:sc_channel}
\end{figure}

%
%
%
%
%
%
%
%

We assume that the array serves the users in the angular range $[-\theta_{\max}, \theta_{\max}]$ for some $\theta_{\max} \in (0,\pi/2)$, and we let 
$d=\frac{\lambda}{2 \sin(\theta_{\max})}$, where $\lambda$ is the wave-length.
%
%
In general, we assume that we can observe only low-dim sketches of the received signal 
via $m \ll M$ linear projections. In the case where the projection matrix contains a single non-zero element equal to 1 in each row, 
we recover the case of array subsampling as a special case. In particular, we shall consider a coprime sampling scheme requiring 
$m = O(2\sqrt{M})$. Coprime subsampling was first developed by Vaidyanathan and Pal in \cite{vaidyanathan2011sparse,vaidyanathan2011theory}, where they showed that for a given spatial span for the array, one obtains approximately the same resolution as a uniform linear array by nonuniformly sampling 
only a few array elements at coprime locations.  We propose several algorithms for estimating the signal subspace and cast them as convex 
optimization problems that can be solved efficiently. 
We also compare via simulation the performance of our algorithms with other state-of-the-art algorithms 
in the literature.  
{The relevance of the proposed approach for JSDM is demonstrated via a representative example, where the dominant subspace of 
users with different channel correlations are estimated and grouped according to the Grassmanian quantization scheme introduced in \cite{nam2014joint}. 
Then, JSDM is applied to the estimated user groups. We compare the achieved sum-rate of our scheme with the ideal case, where
the users' channel covariances are perfectly known, as in \cite{nam2014joint}, and we find that the performance penalty incurred by our proposed method is
negligible, even for very short training lengths.}


\noindent{\bf Notation.}
Throughout the paper, the output of an optimization algorithm $\argmin_{x}  f(x)$ is denoted by $x^*$. 
We use $\bT$ and $\bT_+$ for the space of  all $M\times M$ Hermitian Toeplitz and Hermitian semi-definite Toeplitz matrices. 
We always use $\Id$ for the identity matrix, where the dimension may be explicitly indicated for the sake of clarity (e.g., $\Id_k$ denotes
the $k \times k$ identity matrix). We denote a $k\times k$ diagonal matrix with $k$ diagonal elements $\alpha_1, \dots, \alpha_k$ with $\diag(\alpha_1, \dots, \alpha_k)$.
We define $\bH(M,p)=\{\bfU_{M\times p} \in \bC^{M\times p}: \bfU^\herm \bfU=\Id_p\}$ as the set of tall unitary matrices 
of dimension $M\times p$. For matrices and vectors of appropriate dimensions, 
we define the inner product by $\inp{\bfK}{\bfL}=\tr(\bfK \bfL^\herm)$, where $\tr$ denotes the trace operator, and  define the induced norm by
$\|\bfK\|=\sqrt{\inp{\bfK}{\bfK}}$, also known as {\em Frobenius norm} for matrices. 
For an integer $k \in \ZZ$, we use the shorthand notation $[k]$ 
for the set of non-negative integers $\{0,1, \dots, k-1\}$, where the set is empty if $k < 0$. 


\section{Related Work}\label{sec:lit_review}

Several works in the literature are related to the problem addressed in this paper, which can be summarized in the following four categories:   
Subspace tracking, Low-rank matrix recovery, 
Direction-of-arrival (DoA) estimation,
and Multiple Measurement Vectors (MMV) problem in compressed sensing (CS).
For the sake of completeness, in this section we briefly review these approaches. 

%
While in the rest of this paper we shall treat a general scattering model, 
for simplicity of exposition it is convenient to focus here on a simple model  
in which the transmission between a user and the BS occurs through 
$p$ scatterers (see Fig.~\ref{fig:sc_channel}).  One snapshot of the received signal is given by
\begin{equation} \label{ziofafa}
\bfy=\sum_{\ell=1}^p \bfa(\theta_\ell) w_\ell\,  x + \bfn, 
\end{equation}
where $x$ is the transmitted (training) symbol, $w_\ell \sim \cg(0, \sigma_\ell^2)$ is the channel gain 
of the $\ell$-th multipath component, $\bfn \sim \cg({\bf 0}, \sigma^2\Id_M)$ is the additive white Gaussian noise of the receiver antenna, and where $\bfa(\theta) \in \bC^M$ is the array response at AoA $\theta$, whose $k$-th component is given by 
\begin{align}
[\bfa(\theta)]_k
=e^{j k\pi \frac{\sin(\theta)}{\sin(\theta_{\max})}}.
\end{align}
According to the WSSUS model, the channel gains for different paths, i.e., $\{w_\ell\}_{\ell=1}^p$, 
are uncorrelated.
Without loss of generality, we suppose $x=1$ in all training snapshots. Letting 
$\bfA = [\bfa(\theta_1), \dots,  \bfa(\theta_p)]$, we have 
\begin{align}\label{eq:disc_ch_mod}
\bfy(t)=\bfA \bfw(t) + \bfn(t), \;\;\; t\in[T],
\end{align}
where $\bfw(t) = (w_1(t), \dots, w_p(t))^\transp$ for different $t\in [T]$ are statistically independent. 
Also,  we assume that the AoAs $\{\theta_\ell\}_{\ell=1}^p$ are invariant over the whole training period of length $T$ slots.
%
We gather the received signal's and subsampled signal's snapshots into $M\times T$ 
matrix $\bfY =[\bfy(0), \dots ,\bfy(T-1)]$, and $m\times T$ matrix $\bfX =[\bfx(0), \dots ,\bfx(T-1)]$, 
where $\bfx(t)=\bfB\,\bfy(t)$, and $\bfX=\bfB \bfY$ for the $m\times M$ projection matrix $\bfB$.
From \eqref{eq:disc_ch_mod}, the covariance of $\bfy(t)$ is given by
\begin{align}\label{eq:subs_embed}
{\bfC}_y=\bfA \mathbf{\Sigma} \bfA^\herm + \sigma^2\Id_M= \sum_{\ell=1}^p \sigma_\ell^2 \bfa(\theta_\ell) \bfa(\theta_\ell)^\herm + \sigma^2\Id_M,
\end{align}
where $\Sigmam=\diag(\sigma_1^2, \dots, \sigma_p^2)$ is the covariance matrix of $\bfw(t)$.

\subsection{Subspace Tracking from Incomplete Observations}
Let $\bfC_y = \bfU \Lambdam \bfU^\herm$ be the singular value decomposition (SVD) of $\bfC_y$, where $\Lambdam =\diag(\lambda_1, \dots, \lambda_M)$ denotes the diagonal matrix of singular values (sorted in non-increasing order). 
Denoting by $\bfU_p$ the $M\times p$ matrix consisting of the first $p$ columns of $\bfU$, we have that the columns of $\bfU_p$ form an orthonormal basis 
for the signal subspace. 
The goal of {\em subspace tracking from incomplete observations} consists in estimating 
this subspace from the noisy low-dim sketches $\bfx(t)=\bfB \bfy(t)$, revealed to the estimator sequentially for $t \in [T]$.   
The noiseless version of this problem was studied by Chi et. al. in \cite{chi2013petrels}, proposing 
the PETRELS algorithm.  Another algorithm named GROUSE was proposed by Balzano et. al. in \cite{balzano2010online}.
The main focus of both algorithms is to optimize the computational complexity rather than the data size, and
therefore they are mainly suited to the case where both $M$ and $T$ are high. 

\subsection{Low-rank Matrix Recovery}
For $p \ll M$ and for a high signal-to-noise ratio (SNR), the covariance matrix $\Cm_y$ in (\ref{eq:subs_embed}) is nearly low-rank. 
%
Recovery of low-rank matrices from a collection of a few possibly noisy samples is of great importance in signal processing and machine learning. 
Recently, it has been shown that this can be achieved via nuclear-norm minimization, which is a convex problem and can be efficiently solved \cite{candes2009exact}. 
For a symmetric matrix $\bfM$, the  nuclear norm $\|\Mm\|_*$ is given by the sum of the absolute values of the eigen-values of $\Mm$, 
and reduces to $\tr(\bfM)$ when $\bfM$ is positive semi-definite (PSD). 
In our case, we have only a collection of $T$ snapshots $\bfX=\bfB \bfY$ as defined before. Let
\begin{align}
\widehat{\bfC}_y=\frac{1}{T} \sum_{t=1}^T \bfy(t) \bfy(t)^\herm, \;\;\; \widehat{\bfC}_x=\Bm \widehat{\bfC}_y \Bm^\herm
\end{align}
be the sample covariance of the full and projected signal. 
A natural extension of the matrix completion by nuclear-norm minimization to our case is readily give by:
\begin{align}\label{eq:nuc_norm}
\min_{\bfM} \, \tr(\bfM) \text{ subject to } \bfM \in \bT_+,\  \|\widehat{\bfC}_x- \Bm \bfM \Bm^\herm \| \leq \epsilon,
\end{align}
where $\epsilon$ is an estimate of the $\ell_2$-norm of the error.
%
%

\subsection{Direction-of-arrival Estimation and Super-resolution}

From \eqref{eq:disc_ch_mod}, it is seen that the received signal $\bfy(t)$ is a noisy superposition of $p$ independent Gaussian 
sources arriving from $p$ different angles. 
This is the same model studied for direction-of-arrival (DoA) estimation.  
There are two main categories of algorithms for DoA estimation: classical super-resolution (SR) algorithms such as 
ESPRIT \cite{roy1989esprit} and MUSIC \cite{schmidt1986multiple},  
and more recent compressed sensing based algorithms that use the angular sparsity of the signal 
over a discrete grid of AoAs.
Although grid-based approaches suffer from the mismatch of off-grid sources \cite{chi2011sensitivity}, 
they have been vastly studied \cite{bajwa2010compressed, baraniuk2007compressive, duarte2013spectral, fannjiang2010compressed, herman2009high, malioutov2005sparse, kunis2008random, stoica2012spice, stoica2011new}.
Recently, Cand{\`e}s and Fernandez-Granda  \cite{candes2014towards,candes2013super} 
developed a SR technique based on total-variation (TV) minimization, which 
inherits the convex optimization computational advantage of compressed sensing.
This approach was extended by Tan et. al. in \cite{tan2014direction} to DoA estimation with coprime arrays, when the AoAs 
are sufficiently separated.  In a wireless environment, the AoAs may be clustered. 
This implies that the separation requirement for the SR setup may not be met. 
For example, often  a continuous AoA density function has been observed in measurements (e.g., see \cite{toeltsch2002statistical}), 
and is considered in channel models (e.g., see \cite{asplund2006cost}).
This represents an obstacle for a straightforward 
application of SR methods (both classical \cite{roy1989esprit,schmidt1986multiple} and modern \cite{candes2014towards,candes2013super,tan2014direction}).
Since in this paper we aim at estimating the subspace of the signal rather than DoAs, in Section \ref{sec:SR} we extend the SR 
approach, and develop a new algorithm for estimating the signal subspace.

\subsection{Multiple Measurement Vectors (MMV)}\label{intro:mmv}

%
It is seen from \eqref{eq:disc_ch_mod} that, neglecting the measurement noise $\bfn(t)$, the signal $\bfy(t)$ has typically a sparse representation over the continuous dictionary $\{\bfB \bfa(\theta), \theta \in [-\theta_{\max}, \theta_{\max}]\}$, i.e., only $p$ atoms of the dictionary $\{\bfB \bfa(\theta_i)\}_{i=1}^p$ are needed to represent the signal.
After a suitable discretization of the dictionary (e.g., using a discrete grid of AoAs), the problem of
estimating  $\bfY$ from the collection of snapshots $\bfX=\bfB \bfY$, knowing that each column $\bfx(t)$ of $\bfX$ has the same
sparsity pattern (i.e., it is a linear combination of the same dictionary elements $\{\bfB \bfa(\theta_i)\}_{i=1}^p$) is the classical
Multiple Measurement Vectors (MMV) problem in compressed sensing, which has been widely studied in the literature 
(see e.g., \cite{tropp2006algorithms, tropp2006algorithms2, malioutov2005sparse, lee2012subspace,kim2012compressive, davies2012rank, mishali2008reduce} and refs. therein). 
Since (as in our case) the underlying dictionary may be continuous, more recently off-grid MMV techniques have also been developed  
\cite{tang2013compressed, li2014off, yang2014exact}. 

We will compare the performance of our algorithms with a grid-based MMV approach as in \cite{tropp2006algorithms2},
where the channel coefficients $\bfW=[\bfw(0), \dots, \bfw(T-1)]$ are estimated by \begin{align}\label{eq:l21_optim2}
{\Wm^*}= \argmin_{\Mm \in \bC^{G\times T}} \norm{\Mm}{2,1} \text{ subject to } \norm{\bfX - \bfD \Mm}{} \leq \epsilon,
\end{align}
where $\bfD=[\bfB \bfa(\theta_1), \dots, \bfB \bfa(\theta_G)]$ is a quantized dictionary over a grid of AoAs of size $G$, and where the so-called $\ell_{2,1}$-norm 
of the matrix $\Mm = [\mv_1, \dots, \mv_G]^\transp$  is defined as $\norm{\Mm}{2,1}=\sum_{i=1}^G \|\mv_i\|$, 
where $\bfm_i \in \bC^T$, $i=1,\dots, G$, denote the rows of $\bfM$.
The signal subspace is eventually given by $\sp{\{\bfa(\theta_i): i \in \clA\}}$, where $\clA$ contains 
the index of the ``active'' columns of $\Dm$, i.e., those indexed by the support set of $\bfW^*$. 

We will also compare our algorithms with a grid-less approach inspired by \cite{tang2013compressed, li2014off, yang2014exact}, based on
applying {\em atomic-norm denoising} to the received signal $\bfY$. 
This can be cast as the following semi-definite program (SDP)
\begin{align}\label{eq:atomic_semi}
(\bfT^*, &\bfW^*, \bfZ^*)= \argmin _{\Tm \in \bT_+, \bfW \in \bC^{T \times T}, \bfZ \in \bC^{M\times T} } \tr(\Tm) +  \tr(\bfW)\nonumber\\
&\text{ subject to } \left [ \begin{array}{cc} \Tm & \bfZ\\ \bfZ^\herm & \bfW \end{array} \right ] \succeq {\bf 0}, \norm{\bfX - \bfB\bfZ}{} \leq \epsilon',
\end{align}
where $\Tm^*$ gives an estimate of the signal covariance matrix and
where $\epsilon'$ is an estimate of $\ell_2$-norm of the noise in the projected data. 

In both \eqref{eq:l21_optim2} and \eqref{eq:atomic_semi}, the computational complexity scales with the number of observation snapshots $T$.\footnote{{This scaling depends highly on the specific SDP solver and the structure of the matrix, but it is typically at least of the order $O(T^3)$.}}
This poses a problem  when the training time $T$ is large. Although an ADMM formulation as in \cite{boyd2004convex} is proposed in \cite{li2014off} 
to reduce the computational complexity, the parameters of ADMM need to be selected very carefully to guarantee convergence. 
We shall see that our algorithms perform equally or better than \eqref{eq:l21_optim2} and \eqref{eq:atomic_semi} and have significantly less complexity for large $T$, 
since their complexity does not scale with $T$.

\section{Channel Model and Problem Statement}\label{sec:channel_model}

More general than in (\ref{ziofafa}), the channel vector may be formed by the superposition of a continuum of 
array responses. In order to include this case, we define the AoA {\em scattering function} $\gamma(u)$, 
which describes the received power density along the direction identified by $u \in [-1,1]$, where $u=\frac{\sin(\theta)}{\sin(\theta_{\max})}$ for $\theta\in[-\theta_{\max},\theta_{\max}]$.
We denote the array vector in the $u$ domain by $\bfa(u)$,  where $[\bfa(u)]_k=e^{jk\pi u}$. Then, the channel model is given by 
\begin{align}\label{cont_sig_model}
\bfy(t) = \int_{-1}^1  \sqrt{\gamma(u)} \bfa(u) z(u, t) du+ \bfn(t),
\end{align}
where $z(u,t)$ is a white circularly symmetric Gaussian process with a covariance function 
$\bE\Big [ z(u,t) z(u',t')^* \Big ] = \delta(u - u') \delta_{t,t'}$.
The covariance matrix of $\bfy(t) $ is also given by
\begin{align}
\Cm_y =\int _{-1}^{1} \gamma(u)\bfa(u) \bfa(u)^\herm du+ \sigma^2\bfI_M = \bfS+\sigma^2\bfI_M, \label{C=S+I} 
\end{align}
where $\bfS=\bfS(\gamma):=\int _{-1}^{1} \gamma(u)\bfa(u) \bfa(u)^\herm du$ denotes the covariance matrix of the signal part, 
and where $\sigma^2\bfI_M$ is the covariance matrix of the white additive noise. We define the received SNR by
\begin{align}
\snr:=\frac{\tr(\bfS(\gamma))}{\tr(\sigma^2 \bfI_M)}=\frac{\int_{-1}^1 \gamma(u) \|\bfa(u)\|^2du}{M\sigma^2}=\frac{\int_{-1}^1 \gamma(u) du}{\sigma^2},\nonumber
\end{align}
where $\int_{-1}^1 \gamma(u) du$ is the whole received signal power in a given array element.
For  the ULA, $\bfS$ is a Toeplitz matrix with $[\bfS]_{ij}=[\bff]_{i-j}$, where $\bff$ is an $M$-dim vector 
with $[\bff]_k = \int_{-1}^1 \gamma(u) e^{jk\pi u}\, du$ for $k\in[M]$, 
and corresponds to the $k$-th Fourier coefficient of the density $\gamma$.

We define the best $p$-dim beamforming matrix for the covariance matrix $\bfS$ as $\bfV_p=\argmax _{\bfU \in \bH(M,p)} \inp{\bfS}{\bfU \bfU^\herm}$. Letting $\bfS=\bfV \Lambdam \bfV^\herm=\sum_{i=1}^M \lambda_i \bfv_i \bfv_i^\herm$ be the SVD of $\bfS$, the matrix $\bfV_p$ is an $M\times p$ tall unitary matrix formed by the first $p$ columns of $\bfV$. 
The signal power captured by this beamformer is given by $\inp{\bfS}{\bfV_p\bfV_p^\herm}=\sum_{i=1}^p\lambda_i$.
%
%
%

In this paper, we are concerned with the estimation of $\Vm_p$, for some appropriately chosen $p$, 
from the noisy snapshots of the projected channel (sketches) $\Xm = \Bm \Ym$ obtained during a training period of length $T$, as defined 
at the beginning of Section \ref{sec:lit_review}.
In order to measure the ``goodness'' of estimators, we propose the following performance metric which is relevant 
to the underlying communication problem of JSDM group separation beamforming.
First, we define the efficiency of the best $p$-dim beamformer by
\begin{align}\label{delta_p_def}
\eta_p = \frac{ \inp{\bfS}{\bfV_p \bfV_p^\herm}}{\tr(\bfS)}= \frac{ \tr\big (\bfV_p ^\herm \bfS \bfV_p \big )}{\tr(\bfS)}.
\end{align}
If $\eta_p \approx 1$ for some  $p \ll M$, then a significant amount of signal's power is captured by a low-dim beamformer. 
Let now $\widetilde{\bfV}_p = \widetilde{\bfV}_p(\Xm)$ be an estimator of $\Vm_p$ from the sketches $\Xm$. 
We define the {\em relative efficiency} of $\widetilde{\bfV}_p$ as
\begin{align}\label{eq:perf_metric}
\Gamma_p= \frac{\inp{\bfS}{\widetilde{\bfV}_p \widetilde{\bfV}_p^\herm}}{\inp{\bfS}{\bfV_p  \bfV_p ^\herm}}= 1- \frac{ \inp{\bfS}{\bfV_p  \bfV_p ^\herm}- \inp{\bfS}{\widetilde{\Vm}_p\widetilde{\Vm}_p^\herm} }{\inp{\bfS}{\bfV_p \bfV_p ^\herm}}.
\end{align}
Hence, the efficiency of $\widetilde{\bfV}_p$ is given by $\widetilde{\eta}_p = \Gamma_p \eta_p$, and $1-\Gamma_p$ represents the fraction of signal power lost 
due to the mismatch between the optimal beamformer and its estimate. It is immediate to see that $\Gamma_p \in [0,1]$, where  it is desirable to make it as 
close to $1$ as possible, in particular for those values of $p$ for which $\eta_p \approx 1$.  

\begin{remark} 
We shall compare different subspace estimators for a given channel statistics, number of antennas and number of 
measurements (i.e., $\gamma$, $\snr$, $M$, and  $m$) according to the following procedure: 
1) fix some $\epsilon \in (0,1)$; 2) find minimum $p$ such that $\eta_p \geq 1 - \epsilon$;  3) compare subspace estimators in terms of $\Gamma_p$. 
This approach is quite different from the classical DoA estimation used in array processing (e.g., in radar). 
There, the relevant parameters to be estimated are the AoAs. In our problem, we do not really care about discrete angles, but only about 
a good approximation (in terms of captured signal power) of the span of the corresponding array response vectors. 
It follows that the problem of {\em identifiability} that typically arises in DoA estimation when the minimum angular spacing is too small, 
 is irrelevant here. This is the reason why we can handle continuous AoA scattering functions $\gamma$, in contrast to some SR methods
that assume discrete and sufficiently spaced AoAs. \hfill $\lozenge$
\end{remark}

\section{Proposed Algorithms for Subspace Estimation}

In this section, we introduce the coprime sampling that we will use throughout the paper. 
We explain the proposed algorithms for estimating the signal subspace and provide further intuitions and discussions about their performance. 

\subsection{Coprime Sampling Operator}\label{coprime_subsampling}

Let $\clD$ be a subset of $[M]$ of size $L$ and consider a ULA whose elements are located at $id$ with 
$i \in \clD$ (see Fig.~\ref{fig:sc_channel}). The array is called
a  \textit{minimum-redundancy linear arrays} (MRLA) if for every $\ell \in [M]$, with $\ell\neq 0$, there are unique 
elements $i,i' \in \clD$ such that $\ell=i-i'$. This implies that $M=\frac{L(L-1)}{2}+1$ or approximately $L \approx \sqrt{2M}$. 
Now, consider an arbitrary configuration of sensors $\clD\subset [M]$ and let us define the difference set
\begin{align}
\Delta \clD = \{i-i': i,i' \in \clD \text{ with } i\geq i'\}.
\end{align}
It is clear that $\Delta \clD \subset [M]$. We call $\clD$ a {\em complete cover} (CC) if $\Delta \clD = [M]$. 
This implies that for every $\ell \in [M]$, there is at least a (not necessarily unique) pair $i,i' \in \clD$ such that $\ell=i-i'$. 
By this definition, the location of sensors for a MRLA builds a CC with a minimum size.  
For large values of $M$, it is possible to build a CC of size $2 \sqrt{M}$ by coprime sampling \cite{ vaidyanathan2011sparse,vaidyanathan2011theory}. 
Let $q_1,q_2$ be coprime numbers (i.e., $\gcd(q_1,q_2)=1$) that are very close to each other, and satisfy 
$q_1 q_2 \approx M$, such that $q_1\approx q_2\approx \sqrt{M}$. Let $\clD$ be the set of all nonnegative integer combinations 
of $q_1$ and $q_2$ less than  or equal to $M-1$, i.e.,
$\clD=\cup_{i=1,2} \{k: k\in[M],\, \text{mod}(k,q_i)=0\}$. Note that $|\clD|\approx 2 \sqrt{M}$. 
We define the covering set of an element $k \in [M]$ by
\begin{align}\label{covering_set_k}
\clX_k=\{(i,i'): i,i'\in \clD,\ i\geq i', i-i'=k\},
\end{align}
and its size by $c_k=|\clX_k|$.
For suitable selection of $q_1$ and $q_2$ and for sufficiently large $M$, $c_k\geq 1$ for almost all $k \in [M]$, the set $\Delta \clD$ is approximately equal to $[M]$, and $\clD$ is a CC for $[M]$.\footnote{For small values of $M$, the set $\clD$ might not be a CC, however, as we will explain the performance of our proposed algorithms will not change dramatically as far as the number of uncovered elements in $[M]$ is negligible compared with $M$.} 
In the rest of the paper, we always assume that $\clD$ is a CC for $[M]$.  
Suppose  the elements $d_i$ of $\clD$ are sorted in increasing order with $d_i$ being the $i$-th largest element in the list. 
Also, we let $m=|\clD|$ and let $\bfB$ be the $m \times M$ binary matrix with elements 
 $[\bfB]_{i,d_i}=1$ for $i\in \{1, \dots, m\}$ and zero otherwise. It is immediate to check that $\bfB\bfB^\herm=\bfI_{m}$. 
We will use $\bfB$ as the projection matrix to produces the low-dim observations $\Xm = \Bm\Ym$. 
In passing, this has the advantage that the projection reduces to array subsampling, or ``antenna selection'', which is very easy to implement 
in the analog RF domain by simple switches connecting the selected antennas 
to the RF demodulation chains and A/D converters.  

The first most basic property that a projection matrix $\Bm$ must satisfy in order to allow for efficient subspace estimation is identifiability, 
that is, the associated matrix map must be a bijection when restricted to the class of signal covariance matrices generated by 
the model at hand. For coprime sampling, this is ensured by the following result.

\begin{proposition}\label{prop:bijection}
Let $\bfS$ be an $M\times M$ Hermitian Toeplitz matrix and let $\bfB$ be the coprime sampling matrix. 
Then the mapping $\bfS \to \bfB \bfS \bfB^\herm$ is a bijection.  \hfill $\square$
\end{proposition}
\begin{proof}
Since $\bfS$ is Toeplitz, for any $i, j\in [M]$ with $i\geq j$, we have $[\bfS]_{i,j}=[\bff]_{i-j}$, for some $M$-dim vector $\bff$. Also, as $\bfS$ is Hermitian, $\bff$ fully specifies $\bfS$. Let $i,i' \in \{1, \dots, m\}$ with $i\geq i'$. We can check that
\begin{align}\label{eq:cpr_eq}
[\bfB\bfS \bfB^\herm]_{i,i'}=[\bfS]_{d_i,d_{i'}}=[\bff]_{d_i-d_{i'}}.
\end{align}
As $\clD$ is a complete cover for $[M]$, for any $k \in [M]$ there are $d_i,d_{i'} \in \clD$ such that $d_i-d_{i'}=k$, which using \eqref{eq:cpr_eq} implies that $[\bfB\bfS \bfB^\herm]_{i,i'}=[\bff]_k$. Thus, all the elements of $\bfS$ can be recovered from the low-dim matrix $\bfB\bfS \bfB^\herm$ and vice-versa, thus, the mapping is a bijection.
\end{proof}

\begin{remark}
Although in this paper, for simplicity of implementation, we focus on a coprime sampling matrix $\bfB$, all the proposed algorithms, 
except the super-resolution (\nameSR) algorithm in Section \ref{sec:SR},  can be applied to other sampling matrices (e.g.,  
 i.i.d. Gaussian matrices). \hfill $\lozenge$
\end{remark}

\subsection{\algML: Approximate Maximum Likelihood (\nameML) Estimator}
For the signal model (\ref{cont_sig_model}), we can immediately prove  the following result.
\begin{proposition}\label{prop:suff_stat1}
Let $\widehat{\bfC}_x=\frac{1}{T} \Xm\Xm^\herm$ be the sample covariance of the observations $\Xm$. Then $\widehat{\bfC}_x$ is a sufficient statistics for 
estimating the signal covariance matrix $\Sm$. \hfill $\square$
\end{proposition}

\begin{proof}
Recall that $\bfC_x=\bfB \bfC_y \bfB^\herm=\bfB \bfS \bfB^\herm + \sigma^2\bfI_m$, 
where we have explicitly used the fact that $\Bm\Bm^\herm = \Id_m$. 
As the observations $\bfX$ are Gaussian, after some simple algebra the likelihood function is given by
\begin{align} \label{eq:suff_stat}
p(\Xm|\Sm) &=
\frac{\exp\Big\{-T\, \tr\Big (\widehat{\bfC}_x (\Bm\Sm\Bm^\herm   + \sigma^2\bfI_m)^{-1}\Big )\Big \}}{\pi^{Tm}\, \sdet(\Bm\Sm\Bm^\herm   + \sigma^2\bfI_m)^{T}}.
\end{align}
It follows that the likelihood function depends on $\Xm$ only via $\widehat{\bfC}_x$. 
From the Fischer-Neyman factorization theorem \cite{neyman1936teorema}, 
it follows that  $\widehat{\bfC}_x$ is a sufficient statistics.
\end{proof}

We always assume that the noise variance $\sigma^2$ can be estimated during the system's operation. In this section, for simplicity of the notation, we suppose that the input signal is scaled by $\frac{1}{\sigma}$, and denote the resulting sample covariance by $\widehat{\bfC}_{\widetilde{x}}$, where due to normalization $\widehat{\bfC}_{\widetilde{x}}=\widehat{\bfC}_x/\sigma^2$. Then, the Maximum-Likelihood (ML) estimator for the normalized subsampled data can be written as $\widetilde{\bfS}^* = \argmin_{\widetilde{\bfS} \in \bT_+} L(\widetilde{\bfS})$, where $\widetilde{\bfS}=\bfS/\sigma^2$, and where $L(\widetilde{\bfS})$ is the minus log-likelihood function given by
\begin{align*}
L(\widetilde{\bfS})= \log \sdet(\bfI_m+\bfB \widetilde{\bfS} \bfB^\herm ) + \tr\Big (\widehat{\bfC}_{\widetilde{x}} (\bfI_m+ \bfB \widetilde{\bfS} \bfB^\herm)^{-1} \Big ).
\end{align*}
By direct inspection, we have the following result.
\begin{proposition}\label{prop:ml_cav_vex}
$L(\widetilde{\bfS})$ is the sum of a concave function $\lcav(\widetilde{\bfS})= \log \sdet(\bfI_m+\bfB \widetilde{\bfS} \bfB^\herm )$ and a convex function 
$\lvex(\widetilde{\bfS})= \tr\Big (\widehat{\bfC}_{\widetilde{x}} (\bfI_m+ \bfB \widetilde{\bfS} \bfB^\herm)^{-1} \Big )$.  \hfill $\square$
\end{proposition}


As $L(\widetilde{\bfS})$ is not convex, local optimization techniques such as gradient descent are not guaranteed to converge to the globally optimal solution. 
Since $\widetilde{\bfS}$ scales with SNR, it is possible to obtain a convex (indeed, linear) approximation of the concave function 
$\lcav(\widetilde{\bfS})$, which is tight especially for low SNR. More precisely, we have the following result.  

\begin{proposition}\label{prop:cav_lin_approx}
$\lcav(\widetilde{\bfS})\leq \tr(\bfB\widetilde{\bfS} \bfB^\herm)$ for all $\widetilde{\bfS} \in \bT_+$. 
Moreover, for the low-SNR regime ($\snr \ll 1$), we have $\lcav(\widetilde{\bfS}) = \tr(\bfB\widetilde{\bfS} \bfB^\herm) + o(\snr)$.  \hfill $\square$
\end{proposition}

\begin{proof} 
See Appendix \ref{prop:cav_lin_approx_app}.
\end{proof}

Proposition \ref{prop:cav_lin_approx} states that for low SNR, 
$\tr(\bfB \widetilde{\bfS} \bfB^\herm)$ is the best linear approximation for $\lcav(\widetilde{\bfS})$, which implies that  
\begin{align}\label{eq:aml_1}
L_\text{app}(\widetilde{\bfS})= \tr(\bfB \widetilde{\bfS} \bfB^\herm) + \tr(\widehat{\bfC}_{\widetilde{x}} (\bfI_m+ \bfB \widetilde{\bfS} \bfB^\herm)^{-1}),
\end{align}
is the best convex upper bound for $L(\widetilde{\bfS})$. We define the \textit{approximate maximum likelihood} (AML) estimator for the normalized input by $\widetilde{\bfS}^*=\argmin_{\widetilde{\bfS} \in \bT_+} L(\widetilde{\bfS})$.

\begin{remark}
It is interesting to note that this approximation is valid independent of the length of the training period $T$ as far as $\snr$ is sufficiently small. 
Although the total signal-to-noise ratio of the estimation problem increases by increasing $T$, 
the validity of this approximation only depends on the SNR of an individual sample 
rather than the accumulative signal-to-noise ratio of the whole training samples. On the other hand, increasing $T$ yields
$\widehat{\bfC}_x\rightarrow \Cm_x = \EE[\xv(t)\xv(t)^\herm]$ by consistency of the sample covariance estimator, and improves the estimation.  \hfill $\lozenge$
\end{remark}

The next proposition shows that the AML estimation can be cast as an SDP. 

\begin{proposition}\label{prop:ml_semi_def}
Let $L_\text{app}(\widetilde{\bfS})= \tr(\bfB \widetilde{\bfS} \bfB^\herm) + \tr(\widehat{\bfC}_{\widetilde{x}} (\bfI_m+ \bfB \widetilde{\bfS} \bfB^\herm)^{-1})$ and let $\widehat{\bfC}_{\widetilde{x}}= \bfU \mathbf{\Lambda} \bfU^\herm$ be the SVD of $\widehat{\bfC}_{\widetilde{x}}$. Then the AML estimate can be obtained from the following SDP
\begin{align}\label{eq:ml_semi}
(\widetilde{\bfS}^*, \bfW^*) =& \argmin_{\Mm \in \bT_+, \bfW} \tr(\bfB \Mm \bfB^\herm) + \tr(\bfW)\nonumber\\
& \text{ subject to } \left [  \begin{array}{cc} \bfI_m+ \bfB \Mm \bfB^\herm & \widetilde{\mathbf{\Delta}} \\ \widetilde{\mathbf{\Delta}}^\herm & \bfW \end{array} \right ] \succeq {\bf 0},
\end{align}
where $\widetilde{\mathbf{\Delta}} = \widehat{\bfC}_{\widetilde{x}}^{1/2}=\bfU \mathbf{\Lambda}^{1/2}$.  \hfill $\square$
\end{proposition}

\begin{proof} 
See Appendix \ref{prop:ml_semi_def_app}.
\end{proof}

\noindent Some remarks about optimization problem \eqref{eq:ml_semi} are in order.

\begin{remark}
Although the optimization algorithm \eqref{eq:ml_semi} gives an approximation of the ML estimate for low-SNR regime, it does not need the explicit knowledge of SNR. However, an estimate of noise level in the array is necessary to scale the input data. By Proposition \ref{prop:cav_lin_approx}, we expect that the performance of AML be very close to the performance of the ML 
in the low-SNR regime. \hfill $\lozenge$
\end{remark}

\begin{remark}
After descaling all the parameters by $\sigma$, the SDP in \eqref{eq:ml_semi} can be equivalently written as
\begin{align}\label{eq:ml_semi_descale}
({\bfS}^*,&\bfW^*) = \argmin_{\Mm \in \bT_+, \bfW} \tr(\bfB \Mm \bfB^\herm) + \tr(\bfW)\nonumber\\
& \text{ subject to } \left [  \begin{array}{cc} \sigma^2 \bfI_m+ \bfB \Mm \bfB^\herm & \mathbf{\Delta} \\ \mathbf{\Delta}^\herm & \bfW \end{array} \right ] \succeq {\bf 0},
\end{align}
where $\mathbf{\Delta} = \sigma \widetilde{\mathbf{\Delta}}=\widehat{\bfC}_{x}^{1/2}$, where $\widehat{\bfC}_{x}$ is the sample covariance of the input data without any normalization. This can be used to directly estimate $\bfS$. The optimization \eqref{eq:ml_semi_descale} shows a close resemblance to the 
atomic norm computation introduced in \eqref{eq:atomic_semi} for the MMV problem. 
However, there are also some interesting differences.
For example, in \eqref{eq:ml_semi_descale}, the noise variance $\sigma^2$ and the sampling operator $\bfB$ directly appear in the SDP constraint, whereas in MMV formulation \eqref{eq:atomic_semi}, they appear as an additional regularization term, i.e., $\|\bfX- \bfB \bfT\|\leq \epsilon'$, where $\epsilon'$ can be set by knowing the value of $\sigma$.  
Moreover, the whole observation $\bfX$ during the training period has been replaced by 
$\mathbf{\Delta}=\widehat{\bfC}_{{x}}^{1/2}$ in \eqref{eq:ml_semi_descale}, thus, its complexity is independent of the training length $T$.
\hfill $\lozenge$
\end{remark}

\noindent{\bf Improving AML via \textit{Concave-Convex Procedure}.}  As the cost function $L(\widetilde{\bfS})=\lcav(\widetilde{\bfS})+ \lvex(\widetilde{\bfS})$ is a sum of a convex and a concave function, by slightly modifying algorithm \eqref{eq:ml_semi}, we can obtain better estimates of the signal covariance matrix $\widetilde{\bfS}$ even for high-SNR regime 
via the \textit{concave-convex procedure} (CCCP) \cite{yuille2003concave}. This consists in running a sequence of convex programs iteratively, 
such that in each iteration a better estimate of the optimal (not necessarily the globally optimal) ML solution is computed. 
Let $\widetilde{\bfS}_\ell$, $\ell=0,1,\dots$, denote the  estimate generated at iteration $\ell$. 
Consider the iteration $k$, where the estimates $\widetilde{\bfS}_1, \widetilde{\bfS}_2, \dots, \widetilde{\bfS}_k$ have already been computed. 
Given the last estimate $\widetilde{\bfS}_k$, let $\Gammam_k:=\bfI_m+ \bfB \widetilde{\bfS}_k \bfB^\herm$. 
Then, we can approximate the concave function $\lcav(\widetilde{\bfS})$ as 

\vspace{-2mm}{\small
\begin{align}
\lcav&(\widetilde{\bfS})=\log \sdet(\bfI_m+\bfB\widetilde{\bfS} \bfB^\herm)\nonumber\\
&=\log \sdet (\Gammam_k + \bfB(\widetilde{\bfS} -\widetilde{\bfS}_k)\bfB^\herm)\nonumber\\
&= \log \sdet (\Gammam_k) + \log \sdet \Big(\bfI_m +\Gammam_k^{-1/2} \bfB(\widetilde{\bfS} -\widetilde{\bfS}_k)\bfB^\herm \Gammam_k^{-1/2} \Big )\nonumber\\
&\stackrel{(a)}{\leq} \log \sdet (\bfI+ \bfB \widetilde{\bfS}_k \bfB^\herm) + \tr\Big ( \Gammam_k^{-1/2} \bfB(\widetilde{\bfS} -\widetilde{\bfS}_k)\bfB^\herm \Gammam_k^{-1/2}\Big )\nonumber\\
&=\lcav(\widetilde{\bfS}_k) + \inp{\bfB^\herm \Gammam_k^{-1} \bfB}{\widetilde{\bfS} -\widetilde{\bfS}_k},
\end{align}}%
where in $(a)$, we used an extension of Proposition \ref{prop:cav_lin_approx} proved in Appendix \ref{prop:cav_lin_approx_app} to upper bound 
$\log \sdet (\bfI_m+ \bfH)$ for a Hermitian PSD matrix $\bfH$ by $\trace(\bfH)$.
We also define
\begin{align*}
\Upsilon_k (\widetilde{\bfS};\widetilde{\bfS}_k):=\lcav(\widetilde{\bfS}_k) + \inp{\bfB^\herm \Gammam_k^{-1} \bfB}{\widetilde{\bfS} -\widetilde{\bfS}_k},
\end{align*}
and $L_k(\widetilde{\bfS}; \widetilde{\bfS}_k) := \Upsilon_k (\widetilde{\bfS};\widetilde{\bfS}_k) + \lvex(\widetilde{\bfS})$. It follows that, for arbitrary $\widetilde{\bfS}$ and $\widetilde{\bfS}_k$ in $\bT_+$, 
we have $L(\widetilde{\bfS}) \leq L_k(\widetilde{\bfS}; \widetilde{\bfS}_k)$. Also, from Proposition \ref{prop:cav_lin_approx}, we know that $\Upsilon_k (\widetilde{\bfS};\widetilde{\bfS}_k)$ 
is a tight convex upper bound for the concave function $\lcav(\widetilde{\bfS})$  especially around $\widetilde{\bfS}_k$. Thus, $L_k(\widetilde{\bfS};\widetilde{\bfS}_k)$ is a tight convex upper bound for $L(\widetilde{\bfS})$. 
To find the next estimate  $\widetilde{\bfS}_{k+1}$ in CCCP, we solve the following convex optimization
\begin{align}
\widetilde{\bfS}_{k+1}=\argmin_{\Mm \in \bT_+} L_k(\Mm;\widetilde{\bfS}_k).
\end{align}
Using Proposition \ref{prop:ml_semi_def}, this can also be cast as an SDP that can be efficiently solved. We initialize the estimates with $\widetilde{\bfS}_0={\bf 0}$ for $k=0$. It is immediately seen that $\Upsilon_0 (\widetilde{\bfS};\widetilde{\bfS}_0)=\tr(\bfB\widetilde{\bfS} \bfB^\herm)$, and $L_0(\widetilde{\bfS}, \widetilde{\bfS}_0)=L_\text{app}(\widetilde{\bfS})$ coincides with the AML function in \eqref{eq:aml_1}. Thus, the estimate $\widetilde{\bfS}_1$ corresponds to the AML estimate. We can also see that the sequence of 
estimates $\{\widetilde{\bfS}_k\}_{k=0}^\infty$ monotonically improve the likelihood function, i.e.,
\begin{align}
L(\widetilde{\bfS}_{k+1})&\leq L_k(\widetilde{\bfS}_{k+1};\widetilde{\bfS}_{k})= \min _{\Mm \in \bT_+} L_k(\Mm;\widetilde{\bfS}_k)\\
&\leq L_k(\widetilde{\bfS}_k;\widetilde{\bfS}_k)= L(\widetilde{\bfS}_k),
\end{align}
where we used the identity $\Upsilon_k (\widetilde{\bfS}_k;\widetilde{\bfS}_k)=\lcav(\widetilde{\bfS}_k)$, which implies $L_k(\widetilde{\bfS}_k;\widetilde{\bfS}_k)=L(\widetilde{\bfS}_k)$.
{As a result, if AML is a good approximation of ML, we expect that $\widetilde{\bfS}_1$ provides a good initialization point for the CCCP, such that the sequence $\{\widetilde{\bfS}_k\}_{k=1}^\infty$ converges to the ML estimate (globally optimal point).}

\subsection{\algRMMV: MMV with Reduced Dimension (\nameRMMV)}

One of the main problems with grid-based and off-grid MMV optimizations in \eqref{eq:l21_optim2} and \eqref{eq:atomic_semi} is that their complexity scales very fast with the sample size $T$. Here, we propose an SVD-based technique as in \cite{malioutov2005sparse} to reduce the computational complexity of \eqref{eq:l21_optim2} and \eqref{eq:atomic_semi}. Consider again the observation $\bfX=\bfB \bfY$ during the training period of length $T$. For the time being, 
assume that the discrete AoA model holds and the arrival angles belong to a prefixed grid with elements in the interval 
$[-\theta_{\max},\theta_{\max}]$. In this case, the model $\Xm = \Dm \Wm + \Nm$ with the discretized dictionary 
$\bfD=[\bfB \bfa(\theta_1), \dots, \bfB \bfa(\theta_G)]$ holds exactly. 

We assume that $T \gg m=O(\sqrt{M})$ and consider  ``economy form'' SVD 
$\bfX= \bfU_m \mathbf{\Sigma}_m \bfV_m^\herm$, where $\bfV_m$ is a $T\times m$ tall unitary matrix 
and $\mathbf{\Sigma}_m$ is the $m\times m$ diagonal matrix of the non-zero singular values. 
We define the new data $\widetilde{\bfX}=\bfX \bfV_m= \bfU_m \mathbf{\Sigma}_m$. 
Notice that $\widetilde{\bfX}$ can be simply computed from the sample covariance matrix of the data $\widehat{\bfC}_x= \frac{1}{T} \bfX \bfX^\herm$, 
thus, it is not necessary to store the whole observation $\bfX$ during the training time.  
Moreover, $\widehat{\bfC}_x$ can also be computed from $\widetilde{\bfX}$, thus, Proposition \ref{prop:suff_stat1} 
implies that $\widetilde{\bfX}$ is also a sufficient statistics. We also have 
\begin{align}\label{eq:red_MMV}
\widetilde{\bfX}= \bfD \bfW \bfV_m + \bfN \bfV_m= \bfD \widetilde{\bfW} + \widetilde{\bfN},
\end{align}
where $\widetilde{\bfW}$ and $\widetilde{\bfN}$ of dimension $G \times m$ and $m \times m$ respectively, are the modified channel gains and array noise. 
It is not difficult to check that the reduced observation in \eqref{eq:red_MMV} is still in the MMV format, in the sense that the matrix $\widetilde{\bfW}$ has nonzero rows only on the grid points corresponding to the channel AoAs. However, the dimension of the problem now is fixed and does not scale with $T$. 
Of course, since in reality the AoAs are not exactly placed on a grid, this method suffers from the already mentioned mismatch due to 
domain discretization.

Our second algorithm for subspace estimation, referred to in the following as 
\textit{Reduced MMV} (RMMV), simply applies the off-grid atomic-norm minimization for the MMV problem reviewed in Section \ref{intro:mmv}
to the low-dim data $\widetilde{\bfX}$. This can be cast as the following SDP
\begin{align}
(\bfS^*&, \bfW^*, \bfZ^*)=\argmin _{\Mm \in \bT_+, \bfW \in \bC^{m \times m}, \bfZ \in \bC^{M\times m}} \tr(\Mm) + \tr(\bfW)\nonumber\\ &\text{ subject to } 
\left [ \begin{array}{cc} \Mm & \bfZ\\ \bfZ^\herm & \bfW \end{array} \right ] \succeq {\bf 0},\ \norm{\widetilde{\bfX} - \bfB \bfZ}{} \leq \epsilon,
\end{align}
where $\epsilon$ is an estimate of the norm of $\widetilde{\bfN}$. 
For large values of $T$, we expect that $\widehat{\bfC}_x \approx \sigma^2\bfI_m + \bfB \bfS \bfB^\herm$ by the consistency of the sample covariance estimator, 
such that  the noise components in $\widetilde{\bfN}$ remain approximately independent and Gaussian. 
If $m$ is sufficiently large, then the optimal value of $\epsilon$ concentrates around $\epsilon^* = \sigma \sqrt{m^2} = m \sigma \approx 2\sigma \sqrt{M}$, 
where $\sigma^2$ is the noise variance in each array element, and where we used the fact that, for coprime sampling, 
$m \approx 2\sqrt{M}$. The noise level $\sigma^2$ at the output of the array elements 
can be typically estimated during the system's operation.

\subsection{\algSR: Super Resolution (\nameSR)}\label{sec:SR}
Let $\Sm(\gamma) = \int_{-1}^1 \gamma(u) \av(u)\av(u)^\herm du$ be the signal covariance matrix as in \eqref{C=S+I}. 
It is not difficult to check that, the first column of $\bfS$ contains the vector of Fourier coefficients of $\gamma$ given by $\bff:=\inp{\gamma}{\bfa}:=\int_{-1}^1 \gamma(u) \bfa(u)  du$, where 
$[\bff]_k=[\inp{\gamma}{\bfa}]_k:=\int_{-1}^1 \gamma(u) e^{jk\pi u}  du$, $k\in[M]$.
Since $\bfS(\gamma)$ is Toeplitz, 
Proposition \ref{prop:bijection} yields that for the coprime sampling matrix $\bfB$ 
introduced in Section \ref{coprime_subsampling} all the elements of $\bfS$, and as a result  $\bff$,
can be identified from  $\bfB \bfS \bfB^\herm$. This implies that for a sufficiently large $T$, we can estimate $\bff$ accurately
using the elements of the sample covariance matrix $\widehat{\bfC}_x = \bfB \widehat{\bfC}_y \bfB^\herm$. 
Let $\clX_k$ and $c_k=|\clX_k|$ be the covering set and the covering number of the element $k$, as defined in \eqref{covering_set_k}. 
We define the following estimator for $[\bff]_k$
\begin{align}
[\widehat{\bff}]_k=\frac{\sum_{(i,i')\in \clX_k} [\widehat{\bfC}_x]_{i,i'}}{c_k}.
\end{align}
In Appendix \ref{estim_variance}, in Proposition \ref{estimator_signal}, we prove that for $k\neq0$, $[\widehat{\bff}]_k$ is an unbiased
estimator of $[\bff]_k$ with an approximate variance (exact if $c_k = 1$) given by $\frac{(\sigma^2 + [\bff]_0)^2}{T c_k}$, 
converging to $0$ as $T \rightarrow \infty$.  
We propose the following TV-minimization to recover the subspace of the signal from the estimates $\widehat{\bff}$
\begin{align}\label{eq:tv_pos_meas}
\gamma^* = \argmin \|f\|_{\text{TV}} \text{ subject to } \|\inp{f}{\bfa} - \widehat{\bff}\| \leq \epsilon,
\end{align} 
where $\epsilon$ is an estimate of the norm of the noise in the data. 
For a non-negative measure $\gamma$, the total-variation norm $\|\gamma\|_{\text{TV}}$ is simply given by $\int_{-1}^1 \gamma(u) du= [\bff]_0$. 
Moreover, as $\bff$ can be obtained from the first column of 
the Toeplitz matrix $\bfS(\gamma)$, we can write the optimization \eqref{eq:tv_pos_meas} directly in terms of the covariance matrix: 
\begin{align}\label{eq:tv_pos_meas2}
\bfS^* = &\argmin_{\Mm \in \bT_+} \tr(\Mm) \nonumber\\
&\text{ subject to }  \|\Mm\, \bfe_1 - \widehat{\bff}\| \leq \xi \sqrt{\frac{M}{T}} (\sigma^2 + [\Mm]_{11}),
\end{align}
where $\bfe_1=(1,0,\dots,0)^\transp$ has dimension $M\times 1$, where $[\Mm]_{11}$ is the diagonal element 
of the Toeplitz matrix $\Mm$ (equivalent to $[\bff]_0$), and where the $\epsilon$ parameter in \eqref{eq:tv_pos_meas} 
has been replaced by an estimate thereof in which $\xi\approx 1$ is some parameter that can be tuned appropriately.  
The motivation for \eqref{eq:tv_pos_meas2} is that for sufficiently large $M$ and for the true signal distribution $\gamma$, 
the best value of $\epsilon$ in \eqref{eq:tv_pos_meas} can be estimated by
\begin{align}
\|\inp{\gamma}{\bfa} - \widehat{\bff}\|^2&=\sum_{k} |[\widehat{\bff}]_k - [\bff]_k|^2 \to \sum_k  \bE\Big [\big|[\widehat{\bff}]_k - [\bff]_k\big |^2\Big ]\nonumber\\
&=\sum_k \var\Big[[\widehat{\bff}]_k\Big ]\stackrel{(a)}{\leq} M \frac{(\sigma^2+[\bff]_0)^2}{T},
\end{align}
where in $(a)$ we used the results proved for the variance of the estimate $[\widehat{\bff}]_k$ in Appendix \ref{estimator_signal}, and the fact for those elements with $c_k >1$, the resulting variance is less than $\frac{(\sigma^2+[\bff]_0)^2}{T}$. Thus, we have replaced $\epsilon$ in \eqref{eq:tv_pos_meas} by its approximation
$\sqrt{\frac{M}{T}} (\sigma^2 + [\bff]_0)=\sqrt{\frac{M}{T}} (\sigma^2 + [\bfM]_{11})$, where an additional tuning by the scaling parameter $\xi$ 
has been added to include the variation of this optimal value around its mean.
Algorithm \eqref{eq:tv_pos_meas2} is a convex optimization that can be solved 
if an estimate of the noise variance $\sigma^2$ is available.  In particular, no prior knowledge of SNR is necessary.

\begin{remark}
It might happen, especially for small array size $M$, that some of the elements $k \in [M]$ are not covered by the coprime sampling $\clD$, i.e., $c_k=|\clX_k|=0$. In this case, $[\widehat{\bff}]_k$ can not be estimated for those elements. However, we can still run \eqref{eq:tv_pos_meas} or equivalently \eqref{eq:tv_pos_meas2} by including in the constraint  only the Fourier coefficients corresponding to the elements with $c_k\geq 1$. Note that since the optimization is done over $\bT_+$, if the number of unobserved
elements of $\bff$ (i.e., those elements with $c_k= 0$) is negligible compared with $M$, 
they do not affect the performance considerably.  \hfill $\lozenge$
\end{remark}


\begin{remark}
A coprime sampling scheme similar to ours along with TV-minimization has been used in \cite{tan2014direction} for DoA estimation. 
Provided the AoAs are well-separated, the estimation algorithm in \cite{tan2014direction} can estimate them from the dual optimization proposed in \cite{candes2014towards,candes2013super}. However, the authors do not use the positivity of the measure (in our case $\gamma$) 
that naturally arises because of positive semi-definite property of the signal covariance matrix. Moreover, in this paper, we deal with a wireless 
scattering channel for which the AoAs may be clustered. This implies that the separation requirement for the super-resolution setup may not be met. 
However, since our aim is to  estimate the subspace of the signal rather than AoAs, using the positivity of the underlying measure, we can directly  
solve the primal problem \eqref{eq:tv_pos_meas} rather than the dual one that is used for DoA estimation. In particular, we do not need to go 
through the complicated and error-prone procedure of estimating the support (AoAs) via the dual polynomial, as required 
by DoA estimation in \cite{tan2014direction}.  \hfill $\lozenge$
\end{remark}

\subsection{\algCOR: Covariance Matrix Projection (\nameCMP)}

We define the sampling operator $\sub: \bC^{M\times M} \to \bC^{m \times m}$, 
where $\sub(\bfK)=\bfB\bfK \bfB^\herm$.  
Using the operator $\sub$, we can define a positive semi-definite bilinear form on the space of $M \times M$ matrices 
by $\inpb{\bfK}{\bfL}=\inp{\sub(\bfK)}{\sub(\bfL)}$, where the latter inner product is defined in the space of $m \times m$ matrices.  
This induces the seminorm $\|\bfK\|_\bfB=\sqrt{\inpb{\bfK}{\bfK}}$.\footnote{Note that $\|.\|_\bfB$ is not a norm on the space of all $M\times M$ matrices since we can simply find an $M \times M$ matrix $\bfK\neq 0$ for which $\|\bfK\|_\bfB=0$.} 
It can be easily checked that $\inpb{\bfK}{\bfL}$ restricted to $\bT_+$ is indeed an inner product, and thus it induces a well-defined norm.  We also define 
\begin{align}
\alpha_\bfB(M)=\max_{\bfK \in \bT} \frac{\|\bfK\|}{\|\bfK\|_\bfB},
\end{align}
where dependence on $M$ indicates that the maximization is performed over the space of all $M\times M$ Hermitian Toeplitz matrices. 
The parameter $\alpha_\bfB (M)$ is a measure of coherence of the sampling matrix $\bfB$ with respect to the space of 
Toeplitz matrices. It is not difficult to check that for the coprime matrix $\bfB$, it holds that
$1 \leq \alpha_\bfB(M)\leq \sqrt{M}$.
Our analysis shows that the CMP algorithm, defined in the following, performs better for sampling matrices $\bfB$ with a smaller $\alpha_\bfB(M)$. 

Let $\widehat{\bfC}_x = \frac{1}{T} \Xm\Xm^\herm$. 
In order to recover the dominant $p$-dim subspace of the signal, we first find an estimate of the whole signal covariance matrix $\bfC_y$ by
\begin{align}\label{C_estimator}
\bfC^*_y=\argmin_{\Mm \in \bT_+} \|\widehat{\bfC}_x- \bfB\Mm \bfB^\herm\|.
\end{align}
This is equivalent to the following optimization problem
\begin{align}\label{C_estimator2}
\bfC^*_y=\argmin_{\Mm \in \bT_+} \|\widehat{\bfC}_y - \Mm \|_\bfB,
\end{align}
which implies that the optimal solution $\bfC^*_y$ is given by the projection of the sample covariance matrix of the whole array signal on $\bT_+$ under seminorm $\|.\|_\bfB$. Note that this projection is unique since the restriction of the seminorm $\|.\|_\bfB$ to $\bT_+$ is indeed a norm and the projection theorem holds. 

If an estimate of the noise variance $\sigma^2$ is available, we can directly estimate the covariance matrix of the signal via the following variant of \eqref{C_estimator2}
\begin{align}\label{S_estimator}
\bfS^*=\argmin_{\Mm \in \bT_+} \|\widehat{\bfC}_x- \sigma^2 \bfI_m - \bfB\Mm \bfB^\herm\|.
\end{align}
Once $\bfC^*_y$ or $\bfS^*$ are estimated, we use its $p$-dim dominant subspace 
(for some appropriately chosen $p\in \{1, \dots, M\}$) as an estimate of the signal subspace. 
\section{{Performance Analysis and Discussion}}
In this section, we provide lower bounds on the performance of \nameSR\ and \nameCMP\ algorithms. 
{We did not make progress in the analysis of AML and RMMV due to the difficulty of characterizing the SDP solution.
Therefore, this is left as a future work.}

For the \nameCMP\ Algorithm we obtain the following performance bound.
\begin{theorem}\label{CMP_thm}
Consider the signal model given by \eqref{cont_sig_model} over a training period of length $T$. 
Then, for a given $p \in \{1, \dots, M\}$, the \nameCMP\ estimating a $p$-dim signal subspace 
achieves performance measure $\Gamma_p$  (see (\ref{eq:perf_metric})) satisfying 
\begin{align}
\bE [ \Gamma_p ] &\geq \max \Big \{1- \frac{2\sqrt{2p}}{\eta_p\sqrt{T}} (1+\frac{1}{\snr}),\, 0 \Big \},\\
\Var[ \Gamma_p] &\leq  \frac{8 p}{T \eta_p^2} (1+\frac{1}{\snr})^2,
\end{align}
where $\eta_p$ is defined in \eqref{delta_p_def}, and where $\snr$ denotes the SNR in one snapshot $t\in[T]$.  \hfill $\square$
\end{theorem}

\begin{proof}
See Appendix \ref{sec:cmp_proof}.
\end{proof}

{
\noindent A result similar to Theorem \ref{CMP_thm} holds for the \nameSR\ Algorithm.
\begin{theorem}\label{SR_thm}
Consider the signal model \eqref{cont_sig_model} with the power distribution $\gamma(u)$ with an $M$-dim Fourier coefficients $\bff=\int_{-1}^1 \gamma(u) \bfa(u) du$. Suppose that $\xi$ in \eqref{eq:tv_pos_meas2} is sufficiently large such that $\bff$ is feasible. 
Then, for a given $p \in \{1, \dots, M\}$, the \nameSR\ estimating a $p$-dim signal subspace   
achieves performance measure $\Gamma_p$ satisfying 
\begin{align}
\Gamma_p \geq  1-\frac{4\sqrt{2p} \, \xi}{\eta_p\sqrt{T}} (1+\frac{1}{\snr}),
\end{align}
 where $\snr$ denotes the SNR in one snapshot $t\in[T]$. \hfill $\square$
\end{theorem}
\begin{proof}
See Appendix \ref{sec:cmp_proof}.
\end{proof}
}

\vspace{2mm}
\noindent Some remarks about Theorem \ref{CMP_thm} and \ref{SR_thm} are in order here. 

\begin{remark}
It is seen from Theorem \ref{CMP_thm} that for $T \rightarrow \infty$, $\bE[\Gamma_p]$ converges to $1$ and 
$\var[\Gamma_p]$ tends to $0$, which shows the consistency of \nameCMP. 
{Similarly, it is not difficult to check that for large $T$, by taking $\xi \approx 1$, 
the conditions of Theorem \ref{SR_thm} hold with very high probability. 
This implies that $\lim_{T \rightarrow \infty} \Gamma_p = 1$ in probability, 
yielding the consistency of  \nameSR.} Thus, for large $T$, the $p$-dim subspace estimate obtained from both algorithms 
is as efficient as the best $p$-dim signal subspace.  \hfill $\lozenge$
\end{remark}

\begin{remark}\label{rem:snr}
Both Theorem \ref{CMP_thm} and \ref{SR_thm} indicate that even for infinite SNR, one still needs to take some measurements. 
The reason is that, even in the absence of noise, the signal 
$\bfy(t)$ in \eqref{cont_sig_model} is stochastic and both estimators need some data to discover the underlying signal covariance structure. 
It is also seen that the performance is not very sensitive to SNR when the latter is not too small.
However, for $\snr \downarrow 0$, the required training time for achieving a specific target performance
scales as $T = O(\frac{1}{\snr^2})$. 
{This may indicate that the subspace estimation is quite challenging in low-SNR channels such as mm-wave.}
\hfill $\lozenge$
\end{remark}

\begin{remark}
Let us assume that the signal power is concentrated in an $\alpha$-dim subspace. In this case, $\eta_\alpha\approx 1$, and for a moderate value of SNR, the required training time scales like $T=O(\alpha)$. Two different models can be considered for the signal. In the first model, the effective dimension $\alpha$ does not scale by increasing $M$, thus, the required training length is independent of the embedding dimension or the number of array elements $M$. In the second one, the user has a fixed angular range $\Delta \theta= \beta \pi$, for some $\beta \in (0,1)$,  for which $\alpha \approx \beta M$ scales linearly with $M$. In this case,  the required training length scales linearly in $M$.   \hfill $\lozenge$
\end{remark}

\section{Simulation Results}\label{sec:simulation}

In this section, we assess the performance of our proposed estimators via numerical simulations. The computationally efficient implementation of the proposed optimization problems is beyond the scope of the present work. 
Therefore, here we use the general-purpose CVX package \cite{grant2008cvx} for running all the convex optimizations. 

\subsection{Comparing the Performance of Subspace Estimators}\label{sec:sim_1}
We consider an array of size $M=80$ and $\theta_{\max}=60$ degrees (corresponding 
to an angular sector of $120$ degrees). We use a coprime sampling with $q_1=7, q_2=9$, where we denote the set of indices of the sampled antennas with $\clD$, where $|\clD|=19$.  Although there are still some array indices in $[M]=\{0, \dots, 79\}$ not covered by $\Delta \clD$, the simulations show that the estimators 
are quite insensitive to the presence of a few unobserved elements. We also assume that only $m=20$ RF chains are available at the BS, which would be enough to implement the coprime sampling, and to serve up to $20$ data streams in the UL or DL.

{As an example, we consider a scattering channel with AoAs in the range $\Theta= [\theta_1, \theta'_1] \cup [\theta_2, \theta'_2]$, where $\theta_1=-50$, $\theta'_1=-40$, $\theta_2=10$, and $\theta'_2=20$ degrees. We assume a uniform power  distribution over $\Theta$, thus, the total angular support is $20$ degrees. 
The AoA scattering function $\gamma(u)$ is given by
\begin{align}
\gamma(u)=\left \{ \begin{array}{ll} \frac{\kappa}{\sqrt{1-u^2}} & u\in [u_1, u'_1] \cup [u_2, u'_2]\\ 0 & \text{otherwise,} \end{array} \right.
\end{align}
where $u_i=\sin(\theta_i)$ and $u'_i=\sin(\theta_i')$, for $i=1,2$, and where $\kappa>0$ is a normalization constant. We calculate the vector of Fourier coefficients $\bff=\int_{-1}^1   \gamma(u) \bfa(u) du$, from which we construct the Toeplitz signal covariance matrix $\bfS$. The i.i.d. channel vectors in each training period are generated according to $\bfh=\bfS^{1/2} \bfn_1$, and the corresponding observation
at the array antennas is $\yv = \hv + \sigma\nv_2$, where $\bfn_1, \bfn_2 \sim \cg({\bf 0}, \bfI_M)$ are independent vectors, 
and where $\sigma^2$ denotes the noise variance.}

We compare the performance of each subspace estimator with the optimal beamformer that captures more than $95\%$ of signal's power (i.e., $\eta_p = 0.95$), 
from which we obtain the required dimension $p$.  We estimate the efficiency of each estimator, denoted by  $\Gamma_p$, via numerical simulations. 
To compare the performance of our algorithms with the state of the art, 
we have selected three candidate algorithms that we have also reviewed in Section \ref{sec:lit_review}.
\begin{itemize}
\item {\bf PETRELS.} We have implemented the algorithm introduced in \cite{chi2013petrels} with the difference that 
here we fix the data size. Hence, in every step of the algorithm we select a training sample at random  from the fixed training set and update the estimated signal subspace. 

\item {\bf Nuclear-norm minimization.}
We run the optimization problem given in \eqref{eq:nuc_norm}.
Here we optimistically assume that the algorithm knows the best value of $\epsilon$, given by $\epsilon^*=\|\widehat{\bfC}_x- \bfB \bfC_y \bfB^\herm \|$, although in reality
this would not be available and must be estimated.  

\item {\bf MMV Algorithms.}
We compare our algorithms with the state of the art MMV methods as summarized in Section \ref{intro:mmv}. 
The first algorithm runs the optimization introduced in \eqref{eq:l21_optim2},
where we consider a quantized grid of size $G=3M$ equally spaced AoAs. 
We call the algorithm grid-based MMV (\gb MMV). 
The second one is based on the off-grid techniques given by the optimization \eqref{eq:atomic_semi}.
A byproduct of this optimization is to directly obtain an estimate of the covariance of the data $\Cm^*_y = \Mm^*$, given by the matrix $\Mm^*$ that achieves the minimization
in (\ref{eq:atomic_semi}).
Then, we extract the dominant $p$-dim subspace from such estimate. 
We call this algorithm grid-less MMV (\gl MMV). We set $\epsilon'$ to its optimal value given by $\ell_2$-norm of the noise in subsampled observations $\bfX$. 
\end{itemize}
\begin{figure}[h]
\centering
\includegraphics{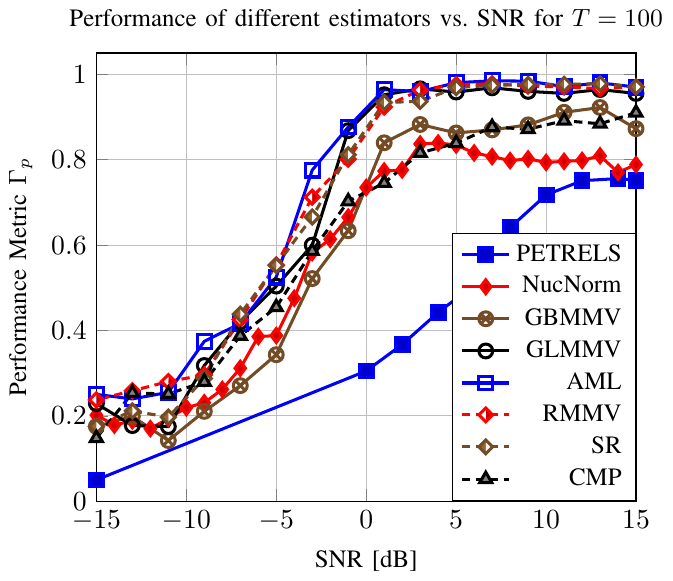}
\caption{Performance comparison of various subspace estimators versus the received SNR for training length $T=100$.}
\label{fig:perf_all_compare_vs_snr}
\end{figure}
\begin{figure}[h]
\centering
\includegraphics{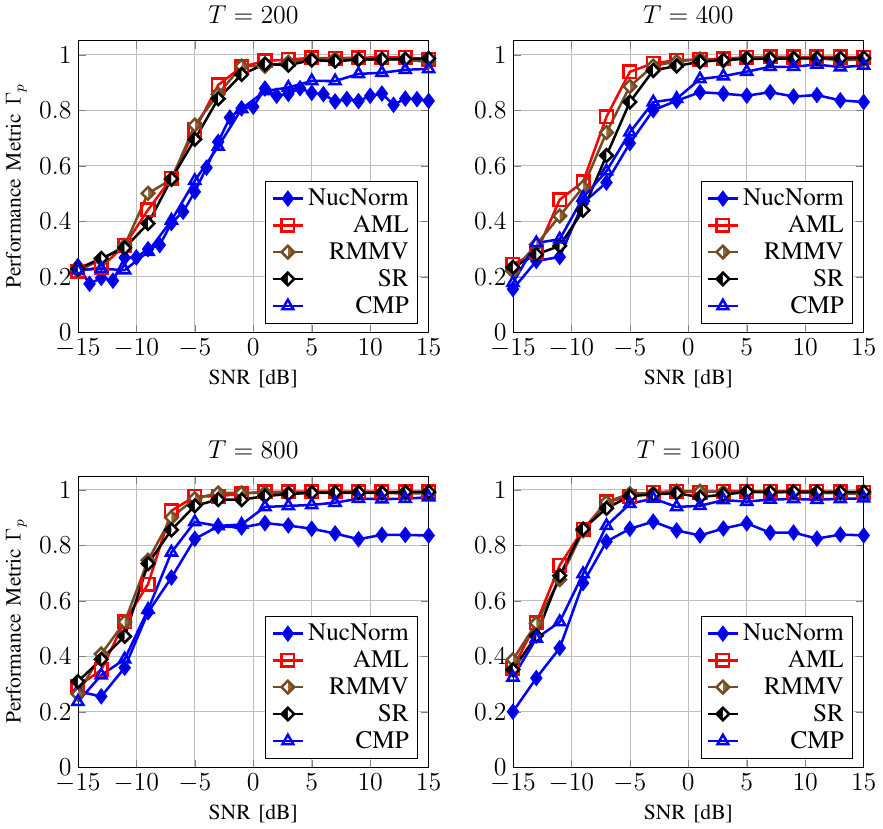}
\caption{{\small Scaling of the performance of different estimators with training length $T\in \{200,400,800,1600\}$}}
\label{fig:perf_vs_T}
\end{figure}

\noindent{\bf Performance vs. signal-to-noise ratio.}
Fig.~\ref{fig:perf_all_compare_vs_snr} compares the performance of our proposed algorithms with the ones in the literature for a range of SNR. {The curve corresponding to each algorithm is obtained by averaging over $20$ independent runs.}
It is seen that \nameML, \nameRMMV, and \nameSR\  perform comparably with the \gl MMV but they have much lower computational complexity, 
which in particular does not scale with $T$. The performance of \nameCMP\ is as good as \gb MMV and better than Nuclear-norm minimization 
especially for higher SNR, but its complexity is much lower than \gb MMV especially for large $T$. 
PETRELS does not perform very well for the fixed data size, e.g., its performance even for $T=800$ is worse than 
the other algorithms.

\noindent {\bf Performance vs. training length $T$.}
Fig.~\ref{fig:perf_vs_T} compares the scaling performance of our proposed algorithms and  Nuclear-norm minimization for different training lengths. As the performance of \nameML\ and \nameRMMV\ is comparable with the \gl MMV and better than \gb MMV and since, for large training length $T$, these algorithms 
run very slowly, we have not included them in this figure. {The results generally show that for moderate range of SNR, the performance of the proposed algorithms improves considerably by increasing $T$. However, for low SNR, the resulting improvement is quite negligible. This confirms the comment made in Remark \ref{rem:snr} about the scaling behaviour of training length $T$ with $\snr$. }

\subsection{JSDM with Subspace Estimation}\label{sec:sim_2}
\noindent{\bf Basic Setup.} As explained in the Introduction, our main motivation for estimating the users' signal subspaces 
comes from JSDM \cite{adhikary2013joint,nam2014joint,adhikary2014joint}, where the users are grouped/clustered based on the similarity of their signal subspaces so that they can be served efficiently by the BS. In this section, we demonstrate the performance of our proposed subspace estimation algorithms included as a component of a JSDM system. In particular, we consider a setup closely inspired by \cite{nam2014joint}, where users have
a uniform AoA scattering function over an interval located in the angular sector  $[-\theta_{\max}, \theta_{\max}]$ covered by the BS. 
Users are grouped by a Grassmannian quantizer with a fixed and pre-defined set of quantization points. 
Following the approach in \cite{nam2014joint}, we fix the Grassmannian quantizer such that each quantization point is a subspace
spanned by a group of adjacent columns of the $M\times M$ {\em Discrete Fourier Transform} (DFT) matrix. Once the users are partitioned into groups, a fixed number of users per group 
is served by JSDM with {\em per-group processing} (PGP) (see details in \cite{adhikary2013joint}). 
The main motivation of this paper is the HDA low-complexity implementation of  JSDM, for which the number of RF chains (and A/D converters) at the BS side is 
limited to $m \ll M$. Therefore, both the estimation of the users' subspaces and the estimation of the per-group effective channels, including 
the JSDM pre-beamforming, must be obtained by sampling no more than $m$ analog signal dimensions. 
In our example, we considere a ULA with $M=80$ elements and $\theta_{\max}=60$ degrees, and the same coprime sampling as in Section \ref{sec:sim_1}, 
and we assume that the BS can only sample $m=20$ analog RF demodulated signals. 
While in \cite{nam2014joint} it is assumed that the users' channel covariance matrices are ideally known, and the grouping by Grassmannian quantization
is performed on the subspaces extracted from the SVD of such covariance matrices, here we estimate the users' subspace using our algorithms, 
from a block of $T$ i.i.d. noisy channel snapshots captured during the training period, as explained before. Then, we apply the same grouping scheme
and DFT pre-beamforming scheme to both the estimated case and the ideally known case, and compare the results in terms of achieved total sum-rate, averaged over a sufficiently large number of independent channel realizations. 

\noindent{\bf Signal Model.} We considered a collection of users $\clK=\{1,\dots, K\}$ of size $K=|\clK|=200$. 
The AoA scattering function of user $k$ is a uniform distribution in the interval $[\theta_k-\frac{\Delta \theta _k}{2}, \theta_k + \frac{\Delta \theta_k}{2}]$, 
corresponding to the following power distribution in the $u$-domain (recall that $u=\frac{\sin(\theta)}{\sin(\theta_{\max})}$):
\begin{align*}
\gamma_k(u)=\left \{ \begin{array}{ll} \frac{\chi}{\sqrt{1-u^2}} & u\in \big[\frac{\sin(\theta_{k}-\frac{\Delta \theta_k}{2})}{\sin(\theta_{\max})}, \frac{\sin(\theta_{k}+\frac{\Delta \theta_k}{2})}{\sin(\theta_{\max})}\big],\\ 0 & \text{otherwise,} \end{array} \right.
\end{align*}
where $\chi>0$ is a normalization constant. 
The users' angular spread (same for all users) is $\Delta \theta_k=\Delta \theta=20$ degress, and the center-AoA $\theta_k$ is i.i.d. and uniformly 
distributed in $[-\theta_{\max}+\frac{\Delta \theta}{2},  \theta_{\max}-\frac{\Delta \theta}{2}]$. 
Letting $\Delta u_k$ indicate the support size of the power distribution $\gamma_k(u)$,
we define $\Delta u = \max_{k \in \clK} \Delta u_k$ as the maximum support size of the users in the $u$-domain, 
where $\Delta u=\frac{2 \sin (\Delta \theta/2)}{\sin(\theta_{\max})} \approx 0.4$.

\noindent{\bf User Grouping.} 
The Grassmanian quantizer is obtained by following  \cite{nam2014joint}.
We divide the domain $u \in [-1,1]$ into intertwined groups $\clG=\{1, \dots, G\}$, with $G=9$, as illustrated in Fig.~\ref{fig:user_cluster}. 
We denote the $u$-support of group $g\in \clG$ with $\clU_g \subset [-1,1]$. Since the length of $\clU_g$ satisfies  $|\clU_g|\approx \Delta u=0.4$, 
due to the intertwined structure of the groups (see Fig.~\ref{fig:user_cluster}), we expect that $\gamma_k(u)$ for each user $k$ is well-aligned with 
at least one $\clU_g$.  
\begin{figure}[h]
\centering
\includegraphics{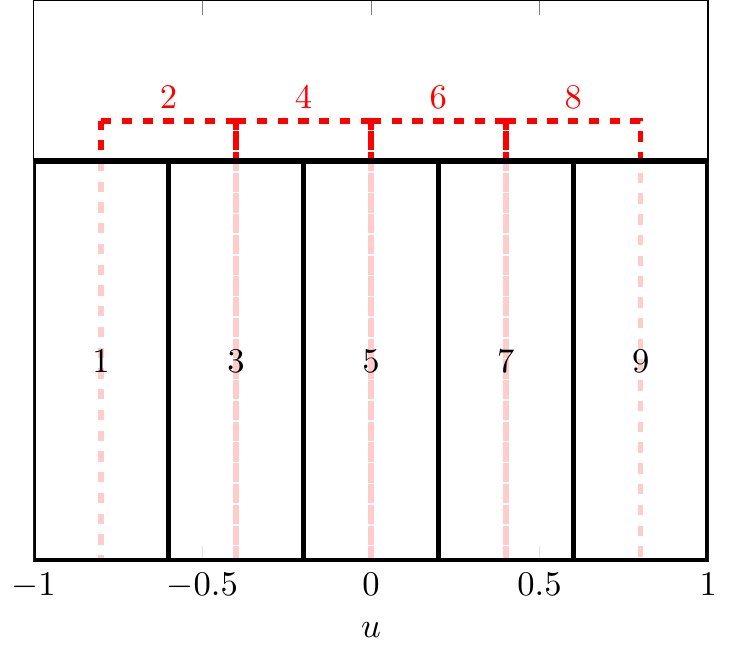}
\caption{Grouping of the users in the $u$-domain.}
\label{fig:user_cluster}
\end{figure}
The signal subspace for group $g \in \clG$ is simply given by a submatrix $\Fm_g$ of the $M \times M$ DFT matrix as follows. 
We define the discrete grid $\clU^\mathsf{disc}=\{-1, -1+\frac{2}{M}, \dots, 1-\frac{2}{M}\}$, and set the columns of $\Fm_g$ to the normalized 
array steering vectors $\frac{1}{\sqrt{M}} \bfa(u)$ for all $u \in \clU^\mathsf{disc} \cap \clU_g$. In our case, 
the matrix $\Fm_g$ has dimension $M\times q_g$ with $q_g=M/5=16$. Therefore,  the JSDM pre-beamforming matrices $\Bm_g$ (see notation in 
\cite{adhikary2013joint}) are simply given by $\Bm_g = \Fm_g$

We cluster the users into groups in $\clG$, where each user $k\in \clK$ is assigned to the group $i_k\in \clG$, where $i_k=\argmax _{g \in \clG} \inp{\bfS_k}{\Fm_g\Fm_g^\herm}$ corresponds to the group whose subspace captures the maximum amount of power in $\bfS_k$. 
This can be seen as an improved weighted version of the chordal distance between subspaces used in  \cite{nam2014joint}, and reduces to 
\cite{nam2014joint} when the non-zero eigenvalues of $\Sm_k$ are constant. 
As said before, we apply the same clustering/quantization scheme both to the ideally known set $\{\Sm_k\}$ and to the estimated set of user covariances
$\{\widehat{\bfS}_k\}$. In this example, we used 
the \nameSR\ Algorithm in \eqref{eq:tv_pos_meas2} with tuning parameter $\xi=2$. Since the performance of our proposed algorithms
in terms of $\Gamma_p$ is similar (see Section \ref{sec:sim_1}), we expect that also the other proposed algorithms yield similar results in terms of JSDM sum-rate
as the ones reported here for \nameSR.

\noindent{\bf Beamforming and Scheduling.}
After clustering the users, JSDM with PGP is applied. Following \cite{nam2014joint},
we divide the groups $\clG$ into two subsets $\clG_1=\{1,3,5,7,9\}$ and $\clG_2=\{2,4,6,8\}$ with intertwined angular supports (see Fig.~\ref{fig:user_cluster}), such that within each subset the groups have mutually orthogonal matrices $\Fm_g$. The two group subsets are served in orthogonal resource blocks, while 
the groups within each subset are served using spatial multiplexing. 

For clarity, we focus on subset $\clG_1$ while the scheme for $\clG_2$ follows immediately. 
Let $\bfh_{g_\ell}$ be the channel vector of the $\ell$-th scheduled user in group $g\in \clG_1$, and let $\bfH_g=[\bfh_{g_1}, \dots, \bfh_{g_{S_g}}]$ 
and $\Hmat=[\bfH_1, \bfH_3, \dots]$ be the channel matrix of the scheduled users in group $g$, and the channel matrix of all groups combined 
in a given coherence block. The BS wishes to serve $S=\sum_{g\in \clG_1} S_g$ data streams (users), where $S_g \leq q_g$ is the number 
of users scheduled in group $g$. We suppose that the scheduled users in group $g$ are selected randomly from the set of users assigned to group $g$ by the grouping algorithm. Since we have only $m=20$ RF chains, we set $S_g=4$, for $g \in \clG_1$ ($S_g=5$ for $g\in \clG_2$). 
The JSDM two-stage precoder is given by $\bfV=\bfU \bfR$, where $\bfU=[\Fm_1, \Fm_3, \dots, \Fm_{G_1}]$ is the $M\times q$ DFT 
pre-beamforming matrix with $q=\sum_{g\in \clG_1} q_g$, and where $\bfR \in \bC^{q\times S}$ is the  baseband (implemented in the digital domain)
MIMO multiuser precoding matrix in block-diagonal form according to the PGP scheme. 
The MIMO precoding matrix $\bfR$ depends on the instantaneous realizations of the reduced dimensional projected 
channel matrices $\Hproj_{gg} = \Fm_g^\herm \bfH_{g}$ for $g \in \clG_1$. The scheduled users in each group are served using linear 
zero-forcing beamforming (ZFBF) matrix  obtained as the column-normalized version of the pseudo-inverse of the projected channel matrix 
$\Hproj_{gg}$. Further details can be found in \cite{adhikary2013joint} and are omitted here due to space limitation.

\begin{figure}[t]
\centering
\includegraphics{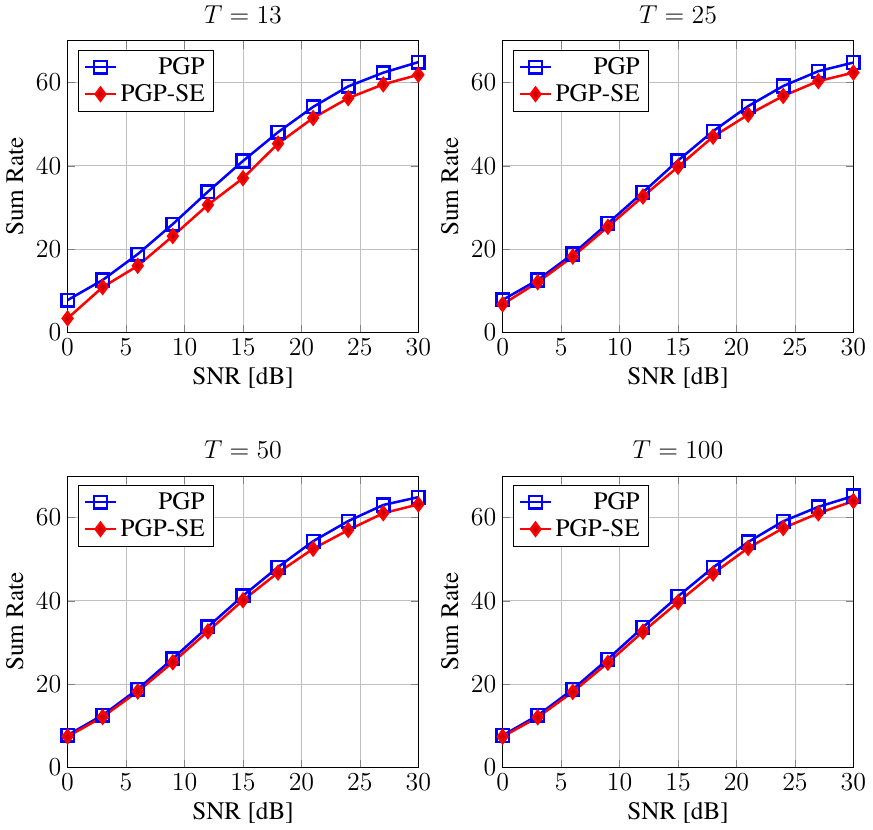}
\caption{Sum-rate vs. SNR performance of a JSDM  with PGP with exact subspace knowledge (PGP), 
and comparison with a JSDM with PGP with subspace estimation (PGP-SE) for training lengths $T \in \{13, 25, 50, 100\}$.}
\label{fig:perf_jsdm_dft}
\end{figure}

\noindent{\bf Sum-rate performance and simulation results.}
Let $\bfU$ and $\bfR$ be the pre-beamforming and the MIMO precoding matrices, and let $\bfd \in \bC^S$ be the $S$-dim vector of $S$ 
symbols corresponding to $S$ data streams. We normalize the symbol energy and the users' channel vectors such that $\bE[\bfd \bfd^\herm]=\bfI_S$ and $\bE[\|\bfh_k\|^2]=1$, for all $k \in \clK$.
When $\bfd$ is transmitted in the DL, the received signal at the user
side is given by $\bfy=\Hmat^\herm \bfU \bfR \bfd + \bfn$, where $\bfn\sim \cg({\bf 0}, \frac{1}{\snr} \bfI_S)$ is the $S$-dim vector of additive noise. 
This is equivalent to an $S \times S$ MIMO multiuser interference channel given by the transfer matrix $\bfT:=\Hmat^\herm \bfU \bfR$. Treating the interference as noise, the \textit{signal-to-interference-plus-noise} (SINR) for the data stream $s$ is given by 
\begin{align}
\sinr_s=\frac{|[\bfT]_{s,s}|^2}{\sum_{s'\neq s} |[\bfT]_{s,s'}|^2 + \frac{1}{\snr}}.
\end{align}
The corresponding achieved sum-rate in the block is given by
\begin{align}
\clC^{\mathsf{sum}}(\Hmat, \snr)=\sum_{s=1}^S \frac{1}{2} \log_2(1+ \sinr_s).
\end{align}
In simulations, we evaluate $\clC^{\mathsf{sum}}(\Hmat, \snr)$ for $1000$ independent realizations of the channel 
matrix $\Hmat$, and plot the average sum-rate, which yields the ``ergodic'' rate achieved via coding over many fading blocks, and ideally adapting the rate of the data streams in each block. 
For simplicity, in our simulations we have used the same value $\snr$ both for DL data transmission and for the UL training, in order to estimate the users' subspaces. In practice, UL and DL SNRs may differ. However, this does not change the qualitative conclusions of our results, since
we can always compensate for a lower UL SNR by increasing the subspace estimation training length $T$. In order to clearly put in evidence 
the effect of the subspace estimation, we calculated the achievable sum-rate under the assumption that the projected channel matrices $\Hproj_{gg}$ are perfectly known at the BS. In general, a dedicated UL training per group must be implemented, from which $\Hproj_{gg}$ is estimated via UL-DL reciprocity \cite{shepard2012argos,rogalin2014scalable}. We did not include here the projected channel estimation phase since 
it incurs a small performance degradation with respect to the ideal case in the practically relevant range of SNR (see analysis in \cite{adhikary2013joint}), 
and would have the effect of ``blurring'' the dependence of the performance with respect to the subspace estimation, which is the focus of this paper. 

Fig.~\ref{fig:perf_jsdm_dft} shows the JSDM sum-rate vs. SNR for the system described before. 
It is seen that subspace estimation even for very short training length $T$ does not incur any significant loss in terms of achievable rate.

\section{Conclusion}
In this paper, we studied the estimation of the user signal subspace in a MIMO wireless scenario where a transmitter (user) with a single antenna 
sends training symbols, and a receiver (BS) equipped with an array with $M$ antennas obtains noisy linear sketches of the corresponding channel vector, modeled as an i.i.d. (in time) and correlated (in the antenna domain) vector Gaussian random process. 
Our algorithms require sampling only $2\sqrt{M}$ antenna array elements and return a $p$-dimensional estimated subspace capturing almost the 
same amount of signal power as the best possible $p$-dimensional subspace.
We analyzed the performance of some of the proposed algorithms and compared it  
with that of state-of-the art methods in the literature via numerical simulations.

{We also illustrated in details how the users' subspace information can be exploited in a JSDM hybrid digital-analog 
implementation of a massive MIMO system. 
We provided numerical simulations showing that in terms of the achieved ergodic sum-rate, such a JSDM system in which the users' subspaces are estimated through
our proposed algorithms performs almost indistinguishably from a scheme that assumes ideal knowledge of the users' channel covariances.
This indicates that the proposed algorithms are well suited for JSDM with HDA implementation with a small number $m \ll M$ of 
RF chains (and A/D converters), 
where the group-separating beamformer is implemented in the analog (RF) domain, 
and the multiuser precoding is implemented in the digital (baseband) domain.} 


\appendices
%
%

\section{Overview of Complex Gaussian Random Variables}\label{app:gaus}
\noindent In this appendix, we review some of the properties of the complex circularly symmetric Gaussian random variables that we need in this paper.

\begin{proposition}\label{prop:circ_gaus}
Let $X=X_r+j X_i$ and $Y=Y_r+j Y_i$ be two zero-mean unit-variance circularly symmetric complex-valued Gaussian  variables with a correlation coefficient $\rho=\rho_r+j \rho_i$. Then we have:
\begin{enumerate}
\item $\bE[XX^*]=\bE[YY^*]=1$, and $\bE[XX]=\bE[YY]=\bE[XY]=0$. 
\item $\bE[X_r^2]=\bE[X_i^2]=\bE[Y_r^2]=\bE[Y_i^2]=\frac{1}{2}$.
\item $\bE[X_rX_i]=\bE[Y_rY_i]=0$, $\bE[X_rY_r]=\bE[X_iY_i]=\frac{\rho_r}{2}$, and $\bE[X_rY_i]=-\bE[X_iY_r]=\frac{\rho_i}{2}$.
\end{enumerate}
\end{proposition}

\noindent We also need the following proposition for real-valued Gaussian variables.
\begin{proposition}[Price's Theorem \cite{price1958useful}]\label{prop:gaus_diff}
Let $Z$ and $W$ be two real-valued $\clN(0,1)$ Gaussian variables with a covariance $\rho$. Let $g(z,w)$ be a differentiable function of $(z,w)$, and $I(\rho)=\bE[g(Z,W)]$. Then $\frac{d}{d\rho}I=\bE\Big [\frac{{\partial}^2}{\partial z \partial w}g(Z,W)\Big ]$. \hfill $\square$
\end{proposition}

\noindent Using Proposition \ref{prop:gaus_diff}, we can prove the following result.

\begin{proposition}\label{prop:circ_gaus2}
Let $X$ and $Y$ be as in Proposition \ref{prop:circ_gaus}. Then, we have
\begin{enumerate}
\item $\bE[X_i^2 Y_r^2]=\bE[X_r^2 Y_i^2]=\frac{1+2 \rho_i^2}{4}$ and $\bE[X_i^2 Y_i^2]=\bE[X_r^2 Y_r^2]=\frac{1+2 \rho_r^2}{4}$.
\item $\bE[|X Y^*|^2]=\bE[|X|^2] \bE[|Y|^2] + |\bE[X Y^*]|^2 $.
\end{enumerate}
\end{proposition}
\begin{proof}
For simplicity, we prove only one of the identities in part 1. Let us consider $\bE[X_r^2 Y_i^2]$. Note that from the properties of the complex Gaussian variables, it results that $(Z,W)=(\sqrt{2}X_r,\sqrt{2}Y_i)$ are jointly Gaussian $\sfN(0,1)$ random variables with covariance $\rho_i$. This implies that $g(\rho_i)=4\bE[X_i^2 Y_r^2]=\bE[Z^2 W^2]$ is a function of $\rho_i$. Applying the Price's theorem in Proposition \ref{prop:gaus_diff}, we have
\begin{align}
\frac{d}{d\rho_i} g= 4\, \bE[Z W]=4\rho_i,
\end{align}
which implies that $g(\rho_i)=2 \rho_i^2 + \kappa$, where $\kappa$ is a constant. For $\rho_i=0$, the random variables $(Z,W)$ are independent from each other, and $g(0)=\bE[Z^2 W^2]=\bE[Z^2]\bE[W^2]=1$, which implies that $\kappa=1$. Hence, we have $g(\rho_i)=2\rho_i^2 +1$, which implies that $\bE[X_i^2 Y_r^2]=\frac{1+2 \rho_i^2}{4}$. 

To prove part 2, note that $|XY^*|^2= (X_r^2+X_i^2)(Y_r^2+Y_i^2)$ can be expanded into four  terms whose expected values can be computed using Price's theorem, where we obtain
\begin{align}\label{eq:gaus_eq1}
\bE[|XY^*|^2]&= 2(\frac{1+2\rho_i^2}{4}) + 2(\frac{1+2 \rho_r^2}{4})=1+|\rho|^2.
\end{align}
Moreover, we have
\begin{align}\label{eq:gaus_eq2}
\bE[XY^*]&=\bE[X_r Y_r + X_i Y_i] + j \bE[X_i Y_r - X_r Y_i]\nonumber\\
&= 2(\frac{\rho_r}{2}) + j 2(-\frac{\rho_i}{2})=\rho_r -j \rho_i=\rho^*.
\end{align}
The result follows from \eqref{eq:gaus_eq1}, \eqref{eq:gaus_eq2} and $\bE[|X|^2]=\bE[|Y|^2]=1$.
\end{proof}


\section{Analysis and Proof Techniques}\label{analysis_proof}

\subsection{Statistics of the Subsampled Signal}\label{estimator_signal}
%
%
From the signal model in \eqref{cont_sig_model}, it is seen that the received signal $\bfy(t)$ is a complex Gaussian  vector with covariance matrix $\bfC_y=\bfS(\gamma)+\sigma^2 \bfI_M$. Since we assume that $\bfy(t)$ is independent across different snapshots $t \in [T]$, this completely specifies its statistics. Similarly, it results that the statistics of the sketches $\bfx(t)=\bfB\bfy(t)$ is fully specified with the covariance matrix $\bfC_x=\bfB \bfC_y \bfB^\herm$. For the coprime sampling matrix $\bfB$, and for $i,j\in \{1,\dots, m\}$ with $i\geq j$, we obtain
\begin{align}\label{eq:sample_x_stat}
[\bfC_x]_{i,j}&=[\bfS(\gamma)]_{d_i,d_j} + \sigma^2 \delta_{ij}\nonumber\\
&=[\bff]_{d_i-d_j} + \sigma^2 \delta_{ij}:=[\bfg]_{d_i-d_j}.
\end{align}
where $d_i \in \clD$ is the $i$-th largest index of the sampled antennas as in Section \ref{coprime_subsampling}, and where $[\bff]_k=\int_{-1}^1 \gamma(u) e^{jk \pi u} du$ denotes the $k$-th Fourier coefficient of $\gamma$. Notice that the SNR is simply given by $\snr=\frac{[\bff]_0}{\sigma^2}$. 

Now consider a specific $k\in [M]$ and let $d_i, d_j \in \clD$ be such that $k=d_i-d_j$. Since, as in Section \ref{coprime_subsampling}, we assume that $\clD$ is a complete cover for $[M]$, such $d_i$ and $d_j$ exist. Let us also define $[\widehat{\bfg}]_k=[\widehat{\bfC}_x]_{i,j}$. The following proposition charactereizes the mean and the variance of $[\widehat{\bfg}]_k$. The proof uses the properties of the complex Gaussian variables reviewed in Appendix \ref{app:gaus}.

\begin{proposition}\label{estim_variance}
Let $k \in [M]$ with $k=d_i-d_j$, and let $[\widehat{\bfg}]_k=[\widehat{\bfC}_x]_{i,j}$. Then, we have
$
\bE\Big [[\widehat{\bfg}]_k\Big ]=[\bfg]_k, \var\Big [[\widehat{\bfg}]_k\Big ]=\frac{(\sigma^2+[\bff]_0)^2}{T}=\frac{\sigma^4(1+\snr)^2}{T}.
$ \hfill $\square$
\end{proposition}

\begin{proof}
Taking the expectation, we have
\begin{align}
\bE\Big [[\widehat{\bfg}]_k\Big ]&=\bE\Big [ [\widehat{\bfC}_x]_{i,j}\Big ]=[\bfC_x]_{i,j}\nonumber\\
&=[\bfC_y]_{d_i,d_j}= [\bff]_{d_i-d_j} + \sigma^2 \delta_{ij}=[\bfg]_k.
\end{align}
Since the observations $\bfy(t)$, and as a result $\bfx(t)$, are independent for $t\in[T]$, and 
{\small
\begin{align}
[\widehat{\bfC}_x]_{i,j}=\frac{1}{T} \sum_{t=1}^T [\bfx(t)]_i [\bfx(t)]_j^*=\frac{1}{T} \sum_{t=1}^T [\bfy(t)]_{d_i} [\bfy(t)]_{d_j}^*,
\end{align}}%
it results that $\var\Big[[\widehat{\bfg}]_k\Big ]= \frac{1}{T} \var \Big [ [\bfy(t)]_{d_i} [\bfy(t)]_{d_j}^*\Big ]$. Hence, using the properties of the complex Gaussian variables proved in Proposition \ref{prop:circ_gaus2}, we obtain
{\small \begin{align}
&\var\Big [ [\bfy(t)]_{d_i} [\bfy(t)]_{d_j}^*\Big ]\nonumber\\
&=\bE\Big [\big |[\bfy(t)]_{d_i} [\bfy(t)]_{d_j}^*\big |^2\Big ] 
- \Big |\bE \Big [ [\bfy(t)]_{d_i} [\bfy(t)]_{d_j}^* \Big ] \Big |^2\nonumber\\
&=\bE\Big [ \big |[\bfy(t)]_{d_i}\big |^2 \Big ] \bE\Big [\big |[\bfy(t)]_{d_j}\big |^2 \Big ] + \Big |\bE \big [ [\bfy(t)]_{d_i} [\bfy(t)]_{d_j}^* \big ] \Big |^2 \nonumber\\
&-\Big |\bE \big [ [\bfy(t)]_{d_i} [\bfy(t)]_{d_j}^* \big ] \Big |^2
=  ([\bff]_0+\sigma^2)^2=\sigma^4(1+\snr)^2,
\end{align}}%
where $\snr=\frac{[\bff]_0}{\sigma^2}$ as before. Hence, $\var\Big [[\widehat{\bfg}]_k\Big ]= \frac{1}{T} \sigma^4(1+\snr)^2$. 
\end{proof}

\subsection{Analysis of the Performance of the \nameCMP\ and \nameSR\ Algorithms}\label{sec:cmp_proof}

In this section, we analyze the performance of the \nameCMP\ and \nameSR. First, we need the following preliminary results.  

\begin{proposition}\label{norm_ineq}
Let $\widehat{\bfC}_y$ be the sample covariance of the signal $\bfy(t)$, $t \in [T]$, and let $\bfC^*_y$ be the \nameCMP\ estimate given by \eqref{C_estimator} or equivalently \eqref{C_estimator2}. Let $\bfC' \in \bT_+$ be an arbitrary Hermitian PSD Toeplitz matrix. Then $\max\Big\{ \|\widehat{\bfC}_y - \bfC^*_y\|_\bfB, \|\bfC^*_y-\bfC'\|_\bfB \Big \} \leq \|\widehat{\bfC}_y- \bfC' \|_\bfB$. \hfill $\square$
\end{proposition}
\begin{proof}
The inequality $ \|\widehat{\bfC}_y - \bfC^*_y\|_\bfB \leq \|\widehat{\bfC}_y- \bfC' \|_\bfB$ simply follows from the definition of $\bfC^*_y$ as the projection of $\widehat{\bfC}_y$ onto the space $\bT_+$. Thus, we only need to prove the other inequality  $\|\bfC^*_y - \bfC'\|_\bfB \leq \|\widehat{\bfC}_ y -\bfC'\|_\bfB$. To prove this, note that the seminorm $\|.\|_\bfB$ is defined from a PSD bilinear form. As $\bfC'$ itself belongs to $\bT_+$, it is not difficult to see that at the projection $\bfC^*_y$, the vector $\bfC'-\bfC^*_y$ is a feasible direction to move because from the convexity of the space $\bT_+$, it results that $\bfC^*_y+ \alpha (\bfC'-\bfC^*_y) \in \bT_+$ for any $\alpha \in [0,1]$. Thus, from the optimality of $\bfC^*_y$, it results that $\inpb{\widehat{\bfC}_y - \bfC^*_y}{\bfC'-\bfC^*_y} \leq 0$. This implies that  
\begin{align}
\|\widehat{\bfC}_y- \bfC' \|_\bfB^2 &= \|\widehat{\bfC}_y- \bfC^*_y+\bfC^*_y-\bfC' \|_\bfB^2\\
&= \|\widehat{\bfC}_y- \bfC^*_y \|_\bfB^2 + \|\bfC^*_y - \bfC'\|_\bfB^2\\
& - 2 \inpb{\widehat{\bfC}_y- \bfC^*_y}{\bfC' - \bfC^*_y}\\
&\geq \|\widehat{\bfC}_y- \bfC^*_y \|_\bfB^2 + \|\bfC^*_y - \bfC'\|_\bfB^2,
\end{align}
from which the desired inequality $\|\bfC^*_y - \bfC'\|_\bfB \leq \|\widehat{\bfC}_y- \bfC'\|_\bfB$ results. Combining the two inequalities, gives the proof.
\end{proof}

\vspace{2mm}
\begin{proposition}\label{prop:concentrate}
Let $\bfC^*_y$ and $\bfC_y$ be as defined before. Suppose $\|\bfC^*_y-\bfC_y\|_\bfB \leq \epsilon$ and let $\bfV\in \bH(M,p)$ be an arbitrary $M\times p$ matrix with $\bfV^\herm \bfV=\bfI_p$. Then, $|\inp{\bfC_y}{\bfV\bfV^\herm}- \inp{\bfC^*_y}{\bfV\bfV^\herm}| \leq  \epsilon \sqrt{p M}$. \hfill $\square$
\end{proposition}
\begin{proof}
Using the Cauchy-Schwartz inequality, we have:
\begin{align}
|\inp{\bfC_y}{\bfV\bfV^\herm}&- \inp{\bfC^*_y}{\bfV\bfV^\herm}|= |\inp{\bfC_y-\bfC^*_y}{\bfV\bfV^\herm}|\nonumber\\
& \leq \|\bfC_y-\bfC^*_y\| \sqrt{\tr\big (\bfV\bfV^\herm \bfV\bfV^\herm \big )}\nonumber\\
&= \|\bfC_y-\bfC^*_y\|\sqrt{p}\nonumber\\
&\stackrel{(a)}{\leq}  \alpha_\bfB(M) \|\bfC_y - \bfC^*_y\|_\bfB \, \sqrt{p}\leq  \epsilon \sqrt{p M},
\end{align}
where in $(a)$ we used the coherence parameter of the coprime matrix $\alpha_\bfB(M) \leq \sqrt{M}$.
\end{proof}

After finding the projection $\bfC^*_y$, we use its $p$-dim dominant subspace to design a beamformer matrix for the received signal $\bfy(t)$, $t\in[T]$. Let
$\bfC^*_y=\bfU \mathbf{\Lambda} \bfU^\herm$ be the SVD of $\bfC^*_y$, and let $\bfV\in \bH(M,p)$ be the matrix consisting of the  first $p$ columns of $\bfU$. If the estimate $\bfC^*_y$ is very close to the $\bfC_y$, then we expect that $\bfV$, in terms of capturing the signal power, be a good approximation of the dominant $p$-dim subspace of the signal. This has been formalized in the following propositions.

\vspace{2mm} 
\begin{proposition}\label{prop:beam_opt}
 Let $\bfC^*_y$ and $\bfC_y$ be as defined before and let $\bfS$ be the signal covariance matrix, where we have $\bfC_y=\bfS + \sigma^2 \bfI_M$. Assume that $\|\bfC^*_y-\bfC_y\|_\bfB \leq \epsilon$. Let $\bfC_y=\bfU\mathbf{\Lambda} \bfU^\herm$ and $\bfC^*_y=\widetilde{\bfU} \widetilde{\mathbf{\Lambda}} \widetilde{\bfU}^\herm$ be the SVD of $\bfC_y$ and $\bfC^*_y$ respectively. Let $\bfV$ and $\widetilde{\bfV}$ be $M\times p$ matrices consisting of the first $p$ columns of $\bfU$ and $\widetilde{\bfU}$. Then, 
$
 |\inp{\bfS}{\bfV\bfV^\herm} - \inp{\bfS}{\widetilde{\bfV}\widetilde{\bfV}^\herm}|\leq 2 \epsilon \sqrt{p M}.
$ \hfill $\square$
\end{proposition}
\begin{proof}
First note that $\bfC_y$ and $\bfS$ differ by a multiple of identity matrix $\bfI_M$. As $\inp{\bfI_M}{\bfV \bfV^\herm}=\inp{\bfI_M}{\widetilde{\bfV}\widetilde{\bfV}^\herm}=p$, we can equivalently prove the following inequality
$
 |\inp{\bfC_y}{\bfV\bfV^\herm} - \inp{\bfC_y}{\widetilde{\bfV}\widetilde{\bfV}^\herm}|\leq 2 \epsilon \sqrt{p M}
$.

From the triangle inequality, $|\inp{\bfC_y}{\bfV\bfV^\herm} - \inp{\bfC_y}{\widetilde{\bfV}\widetilde{\bfV}^\herm}|$ can be upper bounded by 
\begin{align*}
|\inp{\bfC_y}{\bfV\bfV^\herm} - \inp{\bfC^*_y}{{\bfV}{\bfV}^\herm}|+|\inp{\bfC^*_y}{{\bfV}{\bfV}^\herm} - \inp{\bfC_y}{\widetilde{\bfV}\widetilde{\bfV}^\herm}|.
\end{align*}
Using Proposition \ref{prop:concentrate}, the first term is less than $\epsilon \sqrt{pM}$. 
For the second term, note that $0 \leq \inp{\bfC_y}{\widetilde{\bfV}\widetilde{\bfV}^\herm} \leq \inp{\bfC_y}{\bfV\bfV^\herm}$ because from the SVD, $\bfV$ is the best $p$-dim beamformer for $\bfC_y$. Similarly, we have 
$0 \leq \inp{\bfC^*_y}{\bfV\bfV^\herm} \leq \inp{\bfC^*_y}{\widetilde{\bfV}\widetilde{\bfV}^\herm}$.
Consequently, we can always make $|\inp{\bfC^*_y}{{\bfV} {\bfV}^\herm}  -\inp{\bfC_y}{\widetilde{\bfV} \widetilde{\bfV}^\herm} |$ larger by either changing $\bfV$ into $\widetilde{\bfV}$ or vice-versa. This implies that
$|\inp{\bfC^*_y}{{\bfV} {\bfV}^\herm}  -\inp{\bfC_y}{\widetilde{\bfV} \widetilde{\bfV}^\herm} |$ is always smaller than the maximum of $|\inp{\bfC^*_y}{\bfV\bfV^\herm}  -\inp{\bfC_y}{\bfV \bfV^\herm} |$ and $|\inp{\bfC^*_y}{\widetilde{\bfV} \widetilde{\bfV}^\herm}  -\inp{\bfC_y}{\widetilde{\bfV} \widetilde{\bfV}^\herm} |$, where from Proposition \ref{prop:concentrate} both terms are smaller than $\epsilon \sqrt{p M}$. Combining with the first upper bound, we have
\begin{align}\label{eq:max_rel}
|\inp{\bfC^*_y}{\widetilde{\bfV} \widetilde{\bfV}^\herm}  -\inp{\bfC_y}{\bfV \bfV^\herm} | \leq 2\epsilon \sqrt{p M},
\end{align}
which is the desired result.
\end{proof} 

\begin{remark}
Notice that $\bfV$ is the optimal $p$-dim beamformer for $\bfS$ (or equivalently $\bfC_y$). Proposition \ref{prop:beam_opt} implies that the optimal beamformer for the estimate covariance matrix $\bfC^*_y$ is $2\epsilon \sqrt{p M}$-optimal for $\bfS$. \hfill $\lozenge$
\end{remark}
We also need the following result. 
\begin{proposition}\label{c_diff_bc_lemma}
Let $\bfC^*_y$ and $\bfC_y$ be as before. Then 
\begin{align}
\bE\Big [\|\bfC^*_y-\bfC_y\|_\bfB ^2\Big ] \leq \frac{2 M\sigma^4(1+\snr)^2}{T},
\end{align}
which also implies $\bE\Big [\|\bfC^*_y-\bfC_y\|_\bfB\Big ] \leq \frac{\sigma^2 \sqrt{2 M}(1+\snr)}{\sqrt{T}}$. \hfill $\square$
\end{proposition}
\begin{proof}
First notice that, we have 
\begin{align}
\bE\Big [&\|\bfC^*_y-\bfC_y \|_\bfB^2\Big ] \stackrel{(a)}{\leq} \bE\Big [\|\widehat{\bfC}_y-\bfC_y \|_\bfB^2\Big ]\\
&= \sum_{k=0}^{M-1} c_k \bE\Big [ \big |[\widehat{\bfg}]_k - [\bfg]_k\big|^2\Big ]\\
&= \sum_{k=0}^{M-1} c_k \Var\Big [[\widehat{\bfg}]_k\Big ]\stackrel{(b)}{=} \frac{\sigma^4(1+\snr)^2}{T} \sum_{k=0}^{M-1} c_k \\
&= \frac{m(m+1)}{2} \frac{\sigma^4 (1+\snr)^2}{T}\approx \frac{2 M\sigma^4 (1+\snr)^2}{T},
\end{align}
where in $(a)$ we used the inequality proved in Proposition \ref{norm_ineq} by replacing $\bfC'=\bfC_y$, where $[\bfg]_k$ are as in \eqref{eq:sample_x_stat}, where $c_k$ denotes the covering number of $k\in [M]$ by the coprime sampling $\clD$ with $m=|\clD|=O(2 \sqrt{M})$, and where $(b)$ results from Proposition \ref{estim_variance}. The other result simply follows from the identity $\Big \{\bE\Big [\|\bfC^*_y-\bfC_y \|_\bfB\Big ]\Big \}^2\leq \bE\Big [\|\bfC^*_y-\bfC_y \|_\bfB^2\Big ]$.
\end{proof}

\vspace{2mm}
\noindent{\bf Proof of Theorem \ref{CMP_thm}:} Using the definition of $\Gamma_p$ and taking the expectation value we obtain 
\begin{align}
\bE[\Gamma_p] &= 1- \bE \Big [\frac{ \inp{\bfS}{\bfV \bfV^\herm} -\inp{\bfS}{\widetilde{\bfV}\widetilde{\bfV}^\herm} }{\inp{\bfS}{\bfV\bfV^\herm}}\Big ]\\
&\stackrel{(a)}{=}1- \bE \Big [\frac{| \inp{\bfS}{\widetilde{\bfV}\widetilde{\bfV}^\herm} - \inp{\bfS}{\bfV \bfV^\herm} |}{\inp{\bfS}{\bfV\bfV^\herm}}\Big ]\\
&\stackrel{(b)}{=}1- \frac{\bE \big [| \inp{\bfC_y}{\widetilde{\bfV}\widetilde{\bfV}^\herm} - \inp{\bfC_y}{\bfV \bfV^\herm} | \big ]}{\eta_p \tr(\bfS)}\\
&\stackrel{(c)}{\geq} 1- \frac{2 \sqrt{p M}\ \bE \big [\|\bfC^*_y-\bfC_y \|_\bfB\big ]}{\eta_p\, M\, [\bff]_0 } \\
&\stackrel{(d)}{\geq} 1- \frac{2 \sigma^2 \sqrt{2p} M (1+\snr)}{\eta_p\, M\, [\bff]_0 \sqrt{T}} \\
&\geq 1- \frac{2\sqrt{2 p}}{\eta_p\sqrt{T}} (1+\frac{1}{\snr}),
\end{align}
where in $(a)$ we used $\inp{\bfS}{\widetilde{\bfV}\widetilde{\bfV}^\herm} \leq \inp{\bfS}{\bfV \bfV^\herm}$, in $(b)$ we used $\bfC_y=\bfS+ \sigma^2\bfI_M$ and $\tr(\bfV\bfV^\herm)=\tr(\widetilde{\bfV}\widetilde{\bfV}^\herm)$, in $(c)$ we used Proposition \ref{prop:beam_opt}, and finally in $(d)$ we used Proposition \ref{c_diff_bc_lemma}.
As $\Gamma_p \geq 0$, this implies that $\bE\big [\Gamma_p\big ] \geq \max \Big \{1- \frac{2\sqrt{2 p}}{\eta_p\sqrt{T}} (1+\frac{1}{\snr}),\, 0 \Big \}$. In a similar way, we obtain

\begin{align}
\Var[\Gamma_p] &= \Var \Big [ \frac{| \inp{\bfS}{\widetilde{\bfV}\widetilde{\bfV}^\herm} - \inp{\bfS}{\bfV \bfV^\herm} |}{\inp{\bfS}{\bfV\bfV^\herm}} \Big ]\nonumber\\
&= \frac{\Var \Big [ | \inp{\bfC_y}{\widetilde{\bfV}\widetilde{\bfV}^\herm} - \inp{\bfC_y}{\bfV \bfV^\herm} | \Big ]}{\eta_p^2\,  \tr(\bfS)^2}\nonumber\\
&\leq \frac{\bE \Big [ | \inp{\bfC_y}{\widetilde{\bfV}\widetilde{\bfV}^\herm} - \inp{\bfC_y}{\bfV \bfV^\herm} |^2 \Big ]}{\eta_p^2\,  \tr(\bfS)^2}\nonumber\\
&\leq \frac{\bE \Big [ \Big (2 \sqrt{p M} \| \bfC^*_y -\bfC_y \|_\bfB\Big )^2\Big ]}{\eta_p^2\,  \tr(\bfS)^2}\nonumber\\
&= \frac{ 4 p M\ \bE \big [\| \bfC^*_y -\bfC_y \|_\bfB^2\big ] }{\eta_p^2\,  M^2 [\bff]_0^2}\nonumber\\
&\stackrel{(a)}{=} \frac{8 p \, M^2 (\snr+1)^2\sigma^4}{T \eta_p ^2 \, M^2\,  [\bff]_0^2}= \frac{8 p}{T \eta_p^2} (1+\frac{1}{\snr})^2,
\end{align} 
where in $(a)$ we used Proposition \ref{c_diff_bc_lemma} and the fact that $\snr=\frac{[\bff]_0}{\sigma^2}$. This completes the proof.

{
\noindent{\bf Proof of Theorem \ref{SR_thm}}:
Let $\bff$ be the Fourier coefficient of $\gamma$ and let $\widehat{\bff}$ be its estimate as in Section \ref{sec:SR}. Let $\bfS^*$ be Hermitian Toeplitz matrix obtained from the optimization \eqref{eq:tv_pos_meas2} and let us denote its first column by $\bff^*$. Note that $\tr(\bfS^*)=M [\bff^*]_0$. We suppose that the parameter $\xi$ is tuned largely enough such that the true $\bff$ is feasible. Also, from the optimality of $\bff^*$, it results that $[\bff^*]_0 \leq [\bff]_0$. Thus, we have both inequalities
\begin{align}
&\|\bff^*-\widehat{\bff}\|\leq \xi \sqrt{\frac{M}{T}} (\sigma^2+ [\bff^*]_0)\leq \xi \sqrt{\frac{M}{T}} (\sigma^2+ [\bff]_0)\\
&\|\bff-\widehat{\bff}\| \leq \xi \sqrt{\frac{M}{T}} (\sigma^2+ [\bff]_0).
\end{align}
Applying the triangle inequality, we have $\|\bff - \bff^*\| \leq 2\xi \sqrt{\frac{M}{T}} \sigma^2 (1+\snr)$, where as before $\snr=\frac{[\bff]_0}{\sigma^2}$. Using the Toeplitz property, we obtain that 
\begin{align}
\|\bfS-\bfS^*\|\leq \sqrt{2M} \|\bff - \bff^*\|\leq \frac{2\sqrt{2} \xi M \sigma^2}{\sqrt{T}} (1+\snr).
\end{align}
 Let $\bfV$ and $\widetilde{\bfV}$ be $M\times p$ matrices whose columns span the dominant $p$-dim signal subspace of $\bfS$ and $\bfS^*$ respectively. Using a result similar to Proposition \ref{prop:concentrate} and \ref{prop:beam_opt}, we obtain 
\begin{align}
|\inp{\bfS}{\bfV\bfV^\herm}- \inp{\bfS}{\widetilde{\bfV}\widetilde{\bfV}^\herm}|&\leq 2 \sqrt{p} \|\bfS - \bfS^*\|\\
&\leq \frac{4\sqrt{2p} \xi M \sigma^2}{\sqrt{T}} (1+\snr).
\end{align}
Dividing both side by $\inp{\bfS}{\bfV\bfV^\herm}$, using the definition of $\Gamma_p$ in \eqref{eq:perf_metric}, and the definition of $\eta_p$ in \eqref{delta_p_def}, and using the fact that $\Gamma_p \in [0,1]$, we have
\begin{align}\label{final_sr_res}
\Gamma_p \geq  1-\frac{4\sqrt{2p} \, \xi}{\eta_p\sqrt{T}} (1+\frac{1}{\snr}).
\end{align}
This completes the proof.
}


\section{Proofs of the Propositions}

\subsection{Proof of Proposition \ref{prop:cav_lin_approx}}\label{prop:cav_lin_approx_app}
Note that $\bfB \widetilde{\bfS} \bfB^\herm$ is a PSD matrix. We prove the following more general statement that for any $m\times m$  Hermitian matrix $\bfH$ for which $\bfI_m+ \bfH$ is PSD, we have $\log \sdet (\bfI_m+ \bfH)= \tr(\bfH) + o(\tr(\bfH))$. Let $\bfU \mathbf{\Lambda} \bfU^\herm$ be the SVD of $\bfH$. Then, 
\begin{align}
\log \sdet &(\bfI_m+ \bfH)= \log \sdet (\bfI_m+ \bfU \mathbf{\Lambda} \bfU^\herm)\nonumber\\
&= \log \sdet (\bfI_m+  \mathbf{\Lambda})= \sum_{\ell=1}^ m \log(1+ \lambda_\ell)\nonumber\\
&= \sum_{\ell=1}^m \lambda_\ell + o(\tr(\bfH))= \tr(\bfH) + o(\tr(\bfH)).
\end{align}
In particular, from the concavity of the Logarithm, it results that $\log(1+ \lambda_\ell) \leq \lambda_\ell$. This implies that $\log \sdet (\bfI_m+ \bfH) \leq \tr(\bfH)$. This complete the proof.

\subsection{Proof of Proposition \ref{prop:ml_semi_def}}\label{prop:ml_semi_def_app}
{Let $(\widetilde{\bfS}^*, \bfW^*)$ be the output of the SDP \eqref{eq:ml_semi}, and let $\widetilde{\bfS}_\mathsf{opt}$ be the AML estimate given by the optimization  $\widetilde{\bfS}_\mathsf{opt}=\argmin_{\widetilde{\bfS}\in \bT_+} L_\text{app}(\widetilde{\bfS})$. Let us define $\bfH_\mathsf{opt}:=\bfI_m+\bfB \widetilde{\bfS}_\mathsf{opt}\bfB^\herm$ and $\bfW_\mathsf{opt}:=\widetilde{\mathbf{\Delta}}^\herm \bfH_\mathsf{opt}^{-1} \widetilde{\mathbf{\Delta}}$. Using the well-known Schur complement condition for positive semi-definiteness (see \cite{boyd1994linear} p.28), it results that 
$\left [  \begin{array}{cc} \bfH_\mathsf{opt} & \widetilde{\mathbf{\Delta}} \\ \widetilde{\mathbf{\Delta}}^\herm & \bfW_\mathsf{opt}  \end{array} \right ]\succeq {\bf 0}$, 
which implies that $(\widetilde{\bfS}_\mathsf{opt}, \bfW_\mathsf{opt})$ satisfy the SDP constraint in \eqref{eq:ml_semi}. 
Hence,
\begin{align}\label{eq:ml>semi}
L_\text{app}(\widetilde{\bfS}_\mathsf{opt})&= \tr(\bfB \widetilde{\bfS}_\mathsf{opt} \bfB^\herm) + \tr(\widehat{\bfC}_{\widetilde{x}} (\bfI_m+ \bfB \widetilde{\bfS}_\mathsf{opt} \bfB^\herm)^{-1})\nonumber\\
&\stackrel{(a)}{=}\tr(\bfB \widetilde{\bfS}_\mathsf{opt} \bfB^\herm) + \tr(\widetilde{\mathbf{\Delta}}^\herm \bfH_\mathsf{opt}^{-1} \widetilde{\mathbf{\Delta}})\nonumber\\
&=\tr(\bfB \widetilde{\bfS}_\mathsf{opt} \bfB^\herm) + \tr(\bfW_\mathsf{opt}) \nonumber\\
&\geq \tr(\bfB \widetilde{\bfS}^* \bfB^\herm) + \tr(\bfW^*)\nonumber\\
& \stackrel{(b)}{\geq} \tr(\bfB \widetilde{\bfS}^* \bfB^\herm) + \tr(\widetilde{\mathbf{\Delta}}^\herm (\bfI_m+ \bfB \widetilde{\bfS}^* \bfB^\herm)^{-1} \widetilde{\mathbf{\Delta}})\nonumber\\
&= L_\text{app}(\widetilde{\bfS}^*),
\end{align}
where in $(a)$, we use $\widehat{\bfC}_{\widetilde{x}}= \widetilde{\mathbf{\Delta}}\widetilde{\mathbf{\Delta}}^\herm$, and where in $(b)$, we apply Schur complement condition to the SDP constraint to obtain $\bfW^* \succeq \widetilde{\mathbf{\Delta}}^\herm (\bfI_m+ \bfB \widetilde{\bfS}^* \bfB^\herm)^{-1} \widetilde{\mathbf{\Delta}}$, and take the trace of both sides to obtain the desired inequality. As $\widetilde{\bfS}_\mathsf{opt}$ is the AML estimate, we also have $L_\text{app}(\widetilde{\bfS}_\mathsf{opt}) \leq L_\text{app}(\widetilde{\bfS}^*)$, which together with \eqref{eq:ml>semi} completes the proof.
}

\balance
\bibliographystyle{IEEEtran}
{\small
\bibliography{references}}

\end{document}